\documentclass[letterpaper, 10 pt, conference]{ieeeconf} 
\IEEEoverridecommandlockouts

\usepackage{enumitem}
\setlist{topsep=0pt, leftmargin=*}
\usepackage{algorithm}
\usepackage{algorithmicx}
\usepackage{booktabs}

\usepackage[noend]{algpseudocode}
\usepackage{amsmath, amsfonts, bm, amssymb, tabularx}
\usepackage{latexsym, mathtools}
\usepackage{amsthm}
\usepackage{ifthen}
\usepackage{hyperref}\hypersetup{colorlinks=true, unicode=true, linkcolor=[rgb]{0.10,0.05,0.67}, citecolor=[rgb]{0.10,0.05,0.67}, filecolor=[rgb]{0.10,0.05,0.67}, urlcolor=[rgb]{0.10,0.05,0.67}}
\usepackage{Shorthands}
\usepackage{cleveref}
\usepackage{graphicx}
\usepackage{subcaption}
\newcommand{\citep}[1]{\cite{#1}}
\usepackage{wrapfig}
\usepackage{etoolbox}
\usepackage{etoc}

\AtBeginEnvironment{algorithm}{\setlength{\intextsep}{6pt} \setlength{\textfloatsep}{6pt}}
\usepackage{here}

\allowdisplaybreaks

\title{\LARGE \bf
 Policy Gradient over History-Dependent Policy Classes \\ for LQR with Domain Randomization 
}

\author{
Tesshu Fujinami$^1$ , Bruce D. Lee$^2$ , Anastasios Tsiamis$^3$, Nikolai Matni$^1$ , George J. Pappas$^1$ % <-this % stops a space
\thanks{$^1$T. Fujinami, N. Matni and G. J. Pappas are with the Department of Electrical and Systems Engineering, University of Pennsylvania. Emails: {\tt\small\{ftesshu, nmatni, pappasg\}@engineering.upenn.edu} }
\thanks{$^2$B. D. Lee is with the ETH AI Center. Email: {\tt\small\{bruce.lee@ai.ethz.ch\}}}
\thanks{$^3$A. Tsiamis is with the Institute for Automatic Control at ETH Zurich. Email:{\tt\small\{atsiamis@control.ee.ethz.ch\}} }
}%   

\begin{document}

\twocolumn

\maketitle
\thispagestyle{empty}
\pagestyle{empty}

\begin{abstract}
    Domain Randomization (DR) has been widely used to overcome the sim-to-real gap by training a controller on a distribution of simulated environments via reinforcement learning. While DR can achieve robust performance simply using controllers synthesized via policy gradient (PG) methods, the optimization landscape is not well understood, even in the case of linear quadratic regulator (LQR) objectives. To this end, we first study PG of domain randomized LQR over history-dependent policy classes, such as finite impulse response controllers, as they can extend the possibilities of simultaneous stabilization. 
    Second, to find such a stabilizing controller, we propose a curriculum learning based algorithm which gradually expands the memory of the controller. Finally, we show that PG with the proposed algorithm converges globally to the minimizer of a sample average approximation of the DR objective under suitable bounds on the heterogeneity of environments. Empirical results support our findings and highlight promising directions for future work, including nonlinear domain-randomized control.
\end{abstract}
% To this end, we first show that when the heterogeneity gap between sampled systems is large, the simultaneous stabilizing controller requires dynamic parameterization, such as a finite impulse response controller. 

\section{Introduction}

% \bruce{Maybe we can start by directly defining the domain randomization problem to highlight the parts that actually matter, i.e. the control objective, the policy class, and the randomization distribution. }

% Domain randomization (DR) searches for a feedback control policy $K$ belonging to a policy class $\calK$ that minimizes a control objective $J$ on average over system parameters $\Theta$ from a distribution $p_{\Theta}$ \citep{tobin2017dr}:
% \begin{align}
%     \label{eq: domain randomization}
%     K_{DR} = \argmin_{K \in \calK} \E_{\Theta \sim p_{\Theta}}[J(K, \Theta)].
% \end{align}
Domain randomization (DR) searches for a feedback control policy belonging to a policy class that minimizes a control objective on average over system parameters from a distribution \citep{tobin2017dr}.
It is widely used in robot learning because the randomization introduces some robustness to the simulator parameters, thereby enabling sim-to-real transfer in various robotics applications such as manipulation \citep{kuang2026dex4d, deng2025graspvlagraspingfoundationmodel} and locomotion \citep{kumar2021rma}. 
At the same time, the DR policy can be trained by applying standard algorithms from deep reinforcement learning that update policy parameters using policy gradient (PG) methods \citep{fujinami2025policygradientlqrdomain}.  %with gradients estimated via zeroth order information (state transitions and rewards) collected from a simulator. 
Such approaches are parallelizable and can be efficiently accelerated using GPUs. 

Despite the utility of DR, it is not well understood how the choice of control objective, randomization distribution, and policy parameterization impact the efficacy of the resulting policy, and the convergence of policy search algorithms.
\citep{fujinami2025policygradientlqrdomain} present the first analysis of policy gradient for linear quadratic regulator (LQR) with domain randomization over the class of static controllers. However, even in simple examples, the static controller class is insufficient to find a common controller that stabilizes all the systems under the distribution, a property known as simultaneous stabilization \citep{Blondel1993simultaneous}. In this paper, we therefore extend the controller class to the history-dependent policies and study policy gradient for domain randomized LQR over this richer class. We highlight the analytical challenges arising from the history dependence of the controller and establish global convergence by introducing a reparameterized representation of the LQR objective.

\subsection{Related Work}
\paragraph{Robust Control}
The DR problem has previously been proposed in the robust control literature \citep{vidyasagar2001randomized}. It has also been known as
% reduces to a type of control known as
stochastic robust control 
% when the objective $J$ is an indicator for stabilization 
\citep{stengel1991technical, ray1993monte} or average robust control~\cite{bittanti2002iterative}, where the idea is to be robust against  average-case uncertainty. In contrast, approaches like $\calH_\infty$ control \citep{bacsar2008h} and $\mu$-synthesis \citep{packard1993complex} provide robustness for the worst case. The scenario approach \citep{calafiore2006scenario} accounts for the worst case but reduces conservatism, drawing samples to provide probabilistic guarantees. 
% The barrier to adoption for this type of robust control objective control community was twofold: 1) it is difficult to make claims about worst-case performance from optimizing an average objective, and 2) it is typically erroneous to assume that the true world can be perfectly represented by the simulator/model for \emph{any} choice of simulator parameters \citep{zhou1996robust}. In contrast, approaches like $\calH_\infty$ control \citep{bacsar2008h} and $\mu$-synthesis \citep{packard1993complex} provide robustness in the worst case and in the face of unmodeled dynamics. 
% However, such approaches do not scale well to complex nonlinear systems
Improvements in hardware efficiency have enabled the renewed use of gradient-based policy optimization in DR methods, making them practical even for complex nonlinear systems, where worst-case approaches might not scale well. 
% However, such approaches do not scale well to complex nonlinear systems, whereas DR can discover policies that perform well in practice on such systems by using gradient-based policy optimization. In this work, we analyze a setting where such policy optimization approaches converge. 
% \paragraph{Simultaneous Stabilization} DR is related to the problem of simultaneous stabilization~\cite{Blondel1993simultaneous}, where the goal is to find a common controller that stabilizes all systems within a given set, without considering a control cost.

\paragraph{Dual and Adaptive Control}
In this paper, we focus on history-dependent policy classes. If we allow arbitrary policies, then the DR problem becomes similar to the dual/adaptive control problem \citep{feldbaum1960dual, feldbaum1960dual2}.
% which is known to be very hard. 
% In particular, the optimal policy for the DR problem is tasked with minimizing the control objective when it is deployed on an unknown instance sampled from a prior distribution $p_{\Theta}$. Given a controller with sufficient memory, and a control objective with a long enough horizon, the optimal policy would identify the system parameter through interaction, and control that system optimally. 
Solving the dual control problem when the policy class consists of arbitrary history-dependent policies is generally intractable, motivating heuristic approaches based on certainty equivalent synthesis \citep{aastrom1973self, lai1987asymptotically}. Alternatively, we could consider reducing the complexity of the history-dependent policy class, which is the approach we follow in this paper. We focus on the class of linear, stationary, finite impulse response (FIR) policies.
% When the policy class is history-dependent, the domain randomization problem \eqref{eq: domain randomization} also resembles the dual control problem \citep{feldbaum1960dual, feldbaum1960dual2}. In particular, the optimal policy for the DR problem is tasked with minimizing the control objective when it is deployed on an unknown instance sampled from a prior distribution $p_{\Theta}$. Given a controller with sufficient memory, and a control objective with a long enough horizon, the optimal policy would identify the system parameter through interaction, and control that system optimally. Solving the dual control problem when the policy class $\calK$ consists of arbitrary history-dependent policies is generally intractable, motivating heuristic approaches based on certainty equivalent synthesis \citep{aastrom1973self, lai1987asymptotically}. Alternatively, we could consider reducing the complexity of the history dependent policy class $\calK$. In this paper we show a setting where policy gradient methods converge for the objective \eqref{eq: domain randomization} when the policy class is restricted to finite impulse response (FIR) policies.

\paragraph{Domain Randomization}
DR was first used with neural network policies optimized via reinforcement learning in a simulator by \cite{tobin2017dr}, who coined the name ``Domain Randomization''. It has since become the standard approach to encourage robustness in robotic policies trained with reinforcement learning in simulation \citep{loquercio2019deep, peng2018drforcontrol}. The sampling distribution may be a design variable in the synthesis problem \citep{mehta2020active}, or come from a learning procedure \citep{fujinami2025domainrandomizationsampleefficient}, e.g. as the posterior distribution of a Bayesian identification procedure. Recent work has explored generalization of DR in discrete Markov Decision Processes \citep{chen2021understanding, zhong2019pacreinforcementlearningrealworld} and for continuous control \citep{fujinami2025domainrandomizationsampleefficient}. Convergence analysis of policy gradient methods for the LQR problem with domain randomization with static policy classes was presented by \citep{fujinami2025policygradientlqrdomain}. 
% However, such convergence analysis does not consider history-dependent policy classes. 
In this work, we extend the convergence analysis of \citep{fujinami2025policygradientlqrdomain} to history-dependent FIR policies in an LQR setting. 
The use of history dependent policies parametrized by Recurrent Neural Networks such that the resulting policy is “adaptive” rather than just robust was introduced by \citep{akkaya2019solving}.

% Domain randomization, introduced by However, there is a limited understanding for the convergence of policy gradient methods to the optimal solution of the domain randomization problem. Optimization procedures which converge to the optimal solution of the domain randomized linear quadratic regulator (LQR) problem over the set of static state feedback policies have been studied \citep{fujinami2025policygradientlqrdomain}. However, convergence for dynamic policies, has not previously been established. This paper closes this gap by analyzing the convergence of policy gradient over the class of finite impulse response controllers for the LQR problem. 

\subsection{Contribution}

We study the DR problem when the control objective is the LQR objective, and the policy class is the set of FIR linear policies of some order. In this setting:
\begin{itemize}[noitemsep, nolistsep]
     \item We show that policy gradient converges globally under a bounded heterogeneity assumption among systems over the distribution. 
    % \item We present illustrative examples indicating the benefits of  FIR policies over static policies for the DR objective, and the limitations of parameterizing the policy class with any linear policy. 
    \item We provide a curriculum learning approach that gradually expands the memory of the policy during the optimization procedure until satisfactory performance is achieved. 
    \item We present illustrative examples indicating the benefits of  FIR policies over static policies for simultaneous stabilization. We further compare the performance of FIR policies with other parameterizations of dynamic policy and show that FIR policies are not only more amenable to optimization, but also outperform other dynamic policies.
\end{itemize}
Our work establishes a promising direction to leverage insights from robust control, adaptive control, and simultaneous stabilization to build a theoretical grounding for domain randomization which can inform policy parameterization and optimization procedures used in the practice of reinforcement learning for robotics. 
\section{Problem Formulation}

We consider the DR problem in the LQR setting. In particular, suppose we have a fully observable linear time-invariant system of the form
\begin{equation}
    \begin{split}
        x_{t+1} &= A(\theta)x_t + B(\theta)u_t + w_t \label{eq: system}
    \end{split}
\end{equation}
with a parameter $\theta\in\R^{d_\theta}$, state $x_t\in\R^{\dx}$, control input $u_t\in\R^{\du}$, and disturbance $w_t\in\R^{\dx}$ where $w_t\sim\calN(0, W)$ is an i.i.d. Gaussian noise. We assume $W = I$.\footnote{Generalization to arbitrary noise covariance $\Sigma_w\succ0$ may be achieved by a change of basis in state space, which transforms $A(\theta), B(\theta)$, and $Q$ accordingly.} 
For fixed $\theta$, we define the control objective function as
\begin{align}
    J(K,\theta) \coloneqq \lim_{T\rightarrow\infty}\frac{1}{T}\E^K\brac{\sum_{t=0}^T x_t^TQx_t + u_t^TRu_t \middle|\theta},\label{eq: lqr objective}
\end{align}
with $Q\succeq I$ and $R=I$. \footnote{R=I can be enforced by a change of basis in input space} The superscript $K$ indicates that the expectation is taken under control inputs generated by the feedback control policy $K\in\calK$: the policy class we study is made precise in \eqref{eq: FIR controller} and Assumption \ref{asmp: simultaneous stabilization}.
% \begin{figure}[tbp]
%     \centering
%     \includegraphics[width=0.95\linewidth]{figures/optimization_stabilization_tradeoff_plot.png}
%     \caption{Optimization-stabilization tradeoff across controller classes. As controller complexity increases, stabilization performance improves, but the optimization becomes more challenging.}
%     \label{fig: opt-stabilize tradeoff}
% \end{figure}
The goal of domain randomization (DR) is to find a controller that minimizes the expectation of the control objective \eqref{eq: lqr objective} over a distribution of system parameters $\theta$. Let $\Theta$ be a random variable with density $p_\Theta$. Then, the \emph{domain randomized LQR (DR-LQR) problem} is to solve:
\begin{align}
    \label{eq: domain randomization}
    K_{DR} = \argmin_{K \in \calK} \E_{\Theta \sim p_{\Theta}}[J(K, \Theta)].
\end{align}
where we define $\dr(K) \coloneqq \E_{\Theta \sim p_{\Theta}}[J(K, \Theta)]$ as the \emph{domain randomization (DR) objective}. 
%find $K^\star_{DR}\triangleq\argmin_{K\in\calK} \dr(K)$, with
% \begin{align}
%     \dr(K) \triangleq \E_{\Theta\sim p_\Theta} J(K, \Theta). \label{eq: dr objective}
% \end{align}

To tackle this problem, we consider drawing $M$ independent parameter samples,  $\theta_1, \cdots, \theta_M\sim p_\Theta$, to construct an approximation of the DR objective:
\begin{align}
    \sa(K) \coloneqq \frac{1}{M}\sum_{i=1}^M J^i(K), \label{eq: sa objective}
\end{align}
where we denote $J^i(K) \coloneqq J(K, \theta_i)$ and each sample follows the dynamics $x_{t+1} = A^ix_t + B^iu_t + w_t$ with $A^i \coloneqq A(\theta_i), B^i \coloneqq B(\theta_i)$. We define $\sa(K)$ as the \emph{sample average (SA) objective}. 
To minimize the SA objective, we apply policy gradient (PG) methods to the control policy. In particular, starting from a controller $K^{(0)}$, we iteratively update the controller as
\begin{align}
    K^{(n+1)} \leftarrow K^{(n)}- \frac{\alpha}{M}\sum_{i=1}^M\nabla_K J^i(K^{(n)}),\label{eq: sa pg}
\end{align}
where $\alpha$ is a fixed stepsize. Our prior work \cite{fujinami2025policygradientlqrdomain} shows that, under a bounded heterogeneity condition, PG converges globally to the minimizer of \eqref{eq: sa objective} when the policy class is the set of static feedback controllers, $u_t = K_0x_t$, that is, $\calK=\{K_0\in\mathbb{R}^{\du\times \dx}\}$. The limitation of using such a policy class for the DR-LQR problem is illustrated in the following example. 
\begin{example}
    \label{ex: a=1.1,-1.1}
    Consider the scalar system 
    \begin{align}
        x_{t+1} = ax_t + u_t, ~ a \in\curly{-1.1, 1.1}.
    \end{align}
    There does not exist a static linear feedback controller $u_t = kx_t$ that stabilizes both systems simultaneously. This can be shown by considering the closed-loop dynamics
    \[
        x_{t+1} = (a + k)x_t,
    \]
    which is stable if and only if $|a+k|<1$. There does not exist $k$ that satisfies this condition for both systems simultaneously. On the other hand, a policy with one step of memory defined by $u_t = K_0x_t + K_1x_{t-1}$ with $K_0 = 0$ and $K_1 = -0.21$ does simultaneously stabilize both systems. 
\end{example}
The above example indicates that controllers with memory may attain a lower cost for the DR-LQR problem than static controllers. In some cases, such as in the above example, static controllers may fail to even achieve a finite cost, while controllers with memory can.
% ient to achieve a low (or even finite) cost for the domain randomized LQR problem, while controllers with memory can. 
Motivated by this, we consider the \emph{Finite Impulse Response (FIR)} parameterization of the feedback controller. The use of other dynamic parameterizations is discussed in \Cref{s: alternative dynamic}. %, where we argue that FIR controller is simple enough for optimization but still effective for harder scenarios as suggested in \Cref{fig: opt-stabilize tradeoff}. 
% the result relies on the assumption of bounded heterogeneity of the sampled systems. \tas{See comment in the overleaf }To overcome this issue, this paper considers the dynamic policy class and shows that it can deal with systems with a larger heterogeneity gap.
A FIR controller of order $H$ is defined as
\begin{align}
    u_t = \sum_{k=0}^{H-1}K_kx_{t-k} = K\xi_t \label{eq: FIR controller}
\end{align}
where
\begin{align*}
    &K = \begin{bmatrix}
        K_0 & K_1 & \cdots & K_{H-1}
    \end{bmatrix} \in \R^{\du\times H\dx}\\
    &\xi_t = \begin{bmatrix}
    x_t^\top & x_{t-1}^\top & \cdots & x_{t-H+1}^\top
\end{bmatrix}^\top\in\R^{H\dx}
\end{align*}
and $\xi_t$ is the augmented state containing information of the history length $H$. We define $x_t = 0$ for $t \leq 0$ for initialization of the buffer. In the augmented state space, the dynamics can be written as
\begin{align}
        \xi_{t+1} &= \bbA^i\xi_t + \bbB^i u_t + \bbB_w w_t = \bbA_K^i\xi_t + \bbB_w w_t
        \label{eq: fir dynamics}
\end{align}
where 
\begin{align*}
    \bbA^i = \begin{bmatrix}
        A^i & 0 & \cdots & 0 \\
        I & 0 & \cdots & 0 \\
        \vdots & \ddots &  & \vdots \\
        0 & \cdots & I & 0
    \end{bmatrix}, 
    \bbB^i =  \begin{bmatrix}
        B^i \\ 0 \\ \vdots \\ 0
    \end{bmatrix}, 
    \bbB_w =  \begin{bmatrix}
        I \\ 0 \\ \vdots \\ 0
    \end{bmatrix}.
\end{align*}
and $\bbA_K^i \triangleq \bbA^i + \bbB^i K$. 
In the rest of the paper, we define the original system \eqref{eq: system} as the \emph{non-lifted system} and the augmented system \eqref{eq: fir dynamics} as the \emph{lifted system}. 
%\tas{do we need to define the things below ($\calL,\mathbb A^i_K$, etc) here? Can we define them later in the convergence analysis?. Instead I would discuss what the objectives are with respect to this new parameterization.}
In the lifted system, the \emph{FIR-LQR objective} is written as
\begin{align}
    J^i(K) = \trace((\bbQ + K^\top RK)\Sigma_K^i) = \trace(P_K^i\bbW).\label{eq: FIR cost}
\end{align}
where we define the state covariance matrix $\Sigma_K^i$ and the value matrix $P_K^i$ as the solutions to the following Lyapunov equations:
\begin{equation}
    \begin{split}
        &\Sigma_K^i = \bbA_K^i\Sigma_K^i \bbA_K^{i\top} + \bbW \\
        &P_K^i = \bbA_K^{i\top} P_K^i \bbA_K^i + \bbQ + K^\top RK,
    \end{split} \label{eq: fir noise and Q}
\end{equation}
where 
\begin{align*}
    &\bbW = \diag(W, 0, \cdots, 0)\in\R^{H\dx\times H\dx} \\
    &\bbQ = \diag(Q, 0, \cdots, 0)\in\R^{H\dx\times H\dx}.
\end{align*} 

% Let $\bbA_K^i \triangleq \bbA^i + \bbB^i K$ denote the closed-loop matrix of the $i$th sample under an FIR policy of order $H$.
% % \[
% %     \bbA_K^i \triangleq \begin{bmatrix}
% %         A_{K_0}^i & B^iK_1 & \cdots & B^iK_H \\
% %         I & 0 & \cdots & 0 \\
% %         \vdots & \ddots &  & \vdots \\
% %         0 & \cdots & I & 0
% %     \end{bmatrix}, ~ A^i_{K_0} \triangleq A^i + B^iK_0.
% % \]
% When $\bbA_K^i$ is schur stable, we define the state covariance matrix $\Sigma_K^i$ and the value matrix $P_K^i$ as the solutions to the following Lyapunov equations:
% \begin{equation}
%     \begin{split}
%         &\Sigma_K^i = \bbA_K^i\Sigma_K^i (\bbA_K^i)\top + \bbW \\
%         &P_K^i = (\bbA_K^i)^\top P_K^i \bbA_K^i + \bbQ + K^\top RK,
%     \end{split} \label{eq: fir noise and Q}
% \end{equation}
% where 
% \begin{align*}
%     &\bbW = \diag(W, 0, \cdots, 0)\in\R^{H\dx\times H\dx} \\
%     &\bbQ = \diag(Q, 0, \cdots, 0)\in\R^{H\dx\times H\dx}.
% \end{align*} 
% Using these definitions, the LQR objective for system $i$ may be expressed as
% \begin{align}
%     J^i(K) = \trace((\bbQ + K^\top RK)\Sigma_K^i) = \trace(P_K^i\bbW).\label{eq: FIR cost}
% \end{align}
Consequently, \cite[Lemma~1]{fazel2018global} tells us that the gradient of \eqref{eq: FIR cost} is given by
\begin{align}
    \nabla_K J(K, \theta_i) &= 2E_K^i\Sigma_K^i \label{eq: gradient}
\end{align}
where
\begin{align*}
    &E_K^i = \brac{(R + \bbB^{i\top} P_K^i\bbB^i)K + \bbB^{i\top} P_K^i\bbA^i}.
\end{align*}
We can therefore write the gradient of the SA objective as 
\begin{align*}
    \nabla_K \sa(K) = \frac{2}{M} \sum_{i=1}^M E_K^i \Sigma_K^i.
\end{align*}

The goal of this paper is to study the convergence of PG over the class of FIR controllers \eqref{eq: FIR controller} to the optimal solution of the SA objective \eqref{eq: sa objective}. To this end, we show that as long as the sampled systems satisfy a condition on bounded heterogeneity, PG converges to the global optimum starting from an initial stabilizing controller.%The proof consists of two parts. We first show that gradient dominance holds for single LQR objective \eqref{eq: lqr objective} with a FIR controller. Next we introduce the heterogeneity assumption among multiple systems, and prove that PG converges to the minimizer of the sample average surrogate $\sa(K)$ \eqref{eq: sa objective} when samples are close to each other. %We finally characterize the discrepancy between $\sa(K)$ and its population counterpart $\dr(K)$ \eqref{eq: dr objective}. 
\section{Global Convergence of Policy Gradient}
\label{sec: convergence}
% \Tesshu{What to show:}
% \begin{itemize}
%     \item Boundedness on K
%     \item (B,s)-boundedness on J under heterogeneity assumption
%     \item Gradient Domination of SA-LQR by the perturbation argument
% \end{itemize}

Our main result proves the convergence of the PG \eqref{eq: sa pg} over the FIR policy classes  to the optimal solution of the SA objective \eqref{eq: sa objective}. The result requires that the samples are close together, or have a small \emph{heterogeneity gap}. 
To quantify the admissible gap, we use two scalars determined by the problem data: the model size bound $\tau_B \triangleq \max\{1, \sup_{\theta\in\calS}\|B(\theta)\|\}$ and the optimal cost bound $\Jbs \triangleq \sup_{\theta\in\calS}J(K(\theta),\theta)$, where $K(\theta)$ denotes the optimal FIR gain of system $\theta$ and $\calS\subseteq\R^{d_\theta}$ is a compact support of $p_\Theta$. With these quantities, we assume that $\calS$ is small enough that the model error is bounded. 
% The requirement on the heterogeneity gap depends on the order of the FIR policy class, shown in the following assumption. 
% \tasb{I would go for an alternative assumption like the one I shared in slack, see commented text below inside overleaf.}
% \tas{does it make sense to replace the assumption by a lemma+assumption, something like
% \textbf{Lemma:} Let $\theta_1,\dots,\theta_M$. There exist $\het(\theta,H)$, $\grad(\theta,H)$ such that if for any $i,j=1,\dots,M$ model error for any pair$<\inf_k \het(\theta_k,H)$ then gradient domination holds if...... 

% \textbf{Assumption:} The support $S$ is small enough that model error for any pair $<\inf_{\theta \in S}\het(\theta,H)$ and $\inf \grad(\theta,H)>0$.
% }

% \bruce{Agree with Tasos. Let's have the first result be the main statement: there exists $\varepsilon_{\mathsf{het}}(\theta, H)$ such if $\mathsf{diam}(S) \leq \varepsilon_{\mathsf{het}}(\theta, H)$ then $J_{SA}(K^{(N)}) - J_{SA}(K^\star) \leq \varepsilon$. }
\begin{assumption}[Heterogeneity Assumption]
\label{asmp: heterogeneity} 
% Let $\calS\subseteq\R^{d_\theta}$ be such that $p_\Theta(\theta) = 0$ for $\theta\notin\calS$. Let $K^\star$ be a minimizer of $\sa(K)$. There exist $\het,\grad>0$, such that if 
    The support $\calS\subseteq\R^{d_\theta}$ of the randomization distribution $p_\Theta$ satisfies
    \[
      \forall\theta_1, \theta_2\in\calS,\,   \norm{\brac{A(\theta_1) ~ B(\theta_1)} - \brac{A(\theta_2) ~ B(\theta_2)}} \leq \barhet,
    \]
    with the nonzero threshold $\barhet = 1/\fhet(H,\tau_B,\Jbs)$, where $\fhet$ is the explicit function of the history length $H$, the model size bound $\tau_B$, and the optimal cost bound $\Jbs$. 
   % then gradient domination holds
   % \begin{align}
   %  \sa(K) - \sa(K^\star) \leq 4\norm{\nabla_K\sa(K)}_F^2, \label{eq: gradient domination}
   % \end{align}
   % locally for all $K$ such that $\sa(K)\le 8\sa(K^\star)$ and $\norm{\nabla_K\sa(K)}_F \leq \grad$.
\end{assumption}
Since $\barhet$ is explicit, it is checkable on the problem data only and nontrivial because it holds whenever the support $\calS$ is sufficiently small. 
Its role is to guarantee that the SA objective satisfies gradient domination, which is the key to global convergence. See \Cref{s: gradient domination} for further discussion.

In addition, we impose an assumption on the initial controller.
% about the closed-loop stabilizability. 
\begin{assumption}[Initial Stabilization]
    \label{asmp: simultaneous stabilization}
    The initial iterate $K^{(0)}\in\calK$ simultaneously stabilizes the sampled systems, where 
    \begin{align} 
        \calK \coloneqq \curly{K\in\R^{\du\times H\dx} : \rho(\bbA^i + \bbB^iK) < 1, ~ i = 1, \cdots, M}, \label{eq: calK-simultaneous-stable-set}
    \end{align}
    and its cost is within a factor of $8$ of the optimal
    \[
        \zeta_0 \coloneqq \sa(K^{(0)}) \le 8\min_{K\in\calK}\sa(K). 
    \]
\end{assumption}
Later, we introduce a method to obtain an initial stabilizing controller that ensures that this assumption is satisfied. See \Cref{remark: initial stabilization} for further discussion.

Now we introduce the main result of the paper. Under the above heterogeneity assumption and initial stabilization, PG \eqref{eq: sa pg} under a suitable choice of stepsize converges globally to a minimizer of \eqref{eq: sa objective}. Let $K^\star \coloneqq \argmin_{K\in\calK} J_{SA}(K)$, and $K_{\zeta_0}\coloneqq \curly{K\in\calK : \sa(K)\leq \zeta_0}$. 
\begin{theorem}[Global Convergence of Policy Gradient over FIR policy classes]
    \label{thm: gradient dominance of sa-lqr}
    Suppose Assumptions \ref{asmp: heterogeneity} and \ref{asmp: simultaneous stabilization} hold. Consider running \eqref{eq: sa pg} on the FIR policy of order $H$ starting from an initial iterate $K^{(0)}$, and let $L$ be the $L$-smoothness constant
    % \tas{what is the smoothness constant? Lipschitz constant? L-smooth constant? can it be found somewhere? you could also say there exists a step-size such that... } 
    of $\sa(K)$ over $K_{\zeta_0}$
    % \triangleq \curly{K\in\calK : \sa(K)\leq \zeta_0}$. 
    Then policy gradient \eqref{eq: sa pg} with stepsize $\alpha = 1/L$ achieves $\sa(K^{(N)})
    - \sa(K^\star)\leq\epsilon$ after 
    % \tas{it seems somehow that $\varepsilon_{grad}$ affects this result, I cannot tell where it should enter, shouldn't you further restrict $K^0$ to be in $S$?}
    \[
        N\geq 2L\zeta_0\max\curly{1/\cb\epsilon, \fgrad^2}
    \]
    iterations, where $\cb$ and $\fgrad$ are explicit functions of $(H, \tau_B, \Jbs)$. 
    % \tas{are you sure it is poly with H, many terms are exponential in H}.
\end{theorem}
% \tas{not clear where the proof is}
% \tas{Many derivations have bounds that depend on $K$. I think we should get rid of this by using bounded heterogeneity, coercivity, and the condition $J_{SA}(K)\le 8J_{SA}(K^\star)$.}
The proof is sketched out in the rest of this section. 
% Similar to \citep{fazel2018global}, 
The above result provides a convergence guarantee for PG applied to a SA objective. 
% \tas{the following could go to intro}
% Unlike \citep{fazel2018global} and other related problems, 
Unlike \citep{fazel2018global} and other related problems, optimizing the sample average domain randomization problem is an example where there is no known closed-form solution. This makes policy gradient methods a necessary tool, rather than merely a theoretical curiosity.

% Since optimizing the sample average domain randomization problem is an example where there is no known closed-form solution, this result makes policy gradient methods a necessary tool, rather than merely a theoretical curiosity. 

% The result extends the analysis of \citep{fujinami2025policygradientlqrdomain} from considering static policies to considering FIR policies of some order $H$. Example~\ref{ex: a=1.1,-1.1} demonstrates the importance of such generalizations: even for controlling linear systems with a fully observed state, FIR policies may achieve substantially lower costs for the domain randomization problem than static policies. %Due to the strong heterogeneity requirement of Assumption~\ref{asmp: heterogeneity}, this benefit is not fully realized by our theory. 

\begin{remark}[Heterogeneity Assumption]
    As stated in Assumption \ref{asmp: heterogeneity}, we measure open-loop heterogeneity through the A and B matrices. Under this condition, the sampled systems are sufficiently close that a single controller performs well across all of them, which enables global convergence. This heterogeneity requirement might be conservative, as it excludes particular cases in which multiple systems can be simultaneously stabilized despite not being close in the sense of Assumption~\ref{asmp: heterogeneity}. We leave relaxing this condition or determining if it is fundamental to future work.
\end{remark}

\begin{remark}[Initial Stabilization]
    \label{remark: initial stabilization}
    The above convergence guarantee requires that the PG iteration starts from an initial simultaneously stabilizing controller whose cost is within a factor of $8$ of the optimal solution. Consequently, the initial iterate must simultaneously stabilize all sampled systems (as long as such a simultaneously stabilizing controller exists). In \Cref{s: curriculum learning}, we present a scheme to find such an initial stabilizing controller. 
    
% \tas{move the rest of this paragraph to next section? and just say we present a scheme how to find a stabilizing initial controller in the next section?} Section IV of \citep{fujinami2025domainrandomizationsampleefficient} presents a scheme that consists of successively optimizing discounted versions of the sample average objective in which the stage cost of the LQR problems at time $t$ are multiplied by $\gamma^t$ or, equivalently, the state and input matrices for each sample are multiplied by $\sqrt{\gamma}$. The discounting factor $\gamma$ is chosen to gradually increase from $0$ to $1$ such that the discounted versions of the problem always satisfy the condition on the initial iterate. \citep{fujinami2025domainrandomizationsampleefficient} shows that such a scheme is sufficient to ensure convergence to the optimal solution of the sample average objective in the setting of static policies starting from an arbitrary initial iterate. However, the argument is not specific to static policies, and immediately extends to the setting of FIR policies. A related scheme is discussed in \Cref{s: curriculum learning}. 
\end{remark}

% \tas{discuss why heterogeneity should be bounded? (why we need it) (dont discuss the looseness of $\het,\grad$ here, this is discussed later)}

\begin{remark}[Convergence of the DR Objective $\dr(K)$]
The result ensures convergence to the optimal solution of \eqref{eq: sa objective}. However, the original DR objective is \eqref{eq: domain randomization}. To address this discrepancy, concentration arguments can be applied to demonstrate convergence of the SA objective to its population counterpart as the number of sampled systems becomes sufficiently large. The argument of \cite[Thm.~V.1]{fujinami2025policygradientlqrdomain} is not specific to static policies, and extends to FIR policies. 
\end{remark}

% \citep{fazel2018global} establishes the gradient dominance of LQR;the result does not apply to LQR with FIR controller due to the history dependency of the controller.
% Inspired by lifted reparameterization used in LQG analysis such as \citep{fallah2025gradientdominationlqgproblem}, we first prove the gradient dominance in the single system setting. 
Similarly to~\citep{fujinami2025policygradientlqrdomain,fazel2018global}, the convergence proof consists of  three parts: (1) establishing coercivity and smoothness properties of the SA cost, (2) showing that the property known as \emph{gradient domination} \citep{fazel2018global} holds under small heterogeneity, and (3) combining (1) and (2) to show global convergence. 
% The proof of the convergence relies on the extension of 
% a sample average control analysis worked out in 
% \citep{fujinami2025policygradientlqrdomain} under static control policies.
% \tas{can we highlight this further? what makes the analysis challenging compared to your previous work? we could say couple of words in the intro that FIR analysis is challenging}
% Similarly to~\citep{fujinami2025policygradientlqrdomain,fazel2018global}, the convergence proof consists of  three parts: (1) establishing coercivity and smoothness properties of the FIR-LQR cost, (2) showing a property known as \emph{gradient domination} \citep{fazel2018global}, and (3) combining (1) and (2) to show global convergence. 
We highlight that the analysis of \citep{fujinami2025policygradientlqrdomain,fazel2018global} is \emph{insufficient} to prove these properties in the FIR-LQR setting because the LQR cost needs to be analyzed in the lifted state \eqref{eq: fir dynamics} which has a more intricate landscape because of the history dependence of the controller $K$ \eqref{eq: FIR controller}. In particular, the noise covariance and state penalties are singular in the lifted state space~\eqref{eq: fir noise and Q}.
% \tas{we need to explicitly state that previous analysis is insufficient to prove smoothness,coercivity,etc.. and requires the alternative analysis here}. 
% In the remainder of this section, we outline the main steps to derive the extension. 
In Section \ref{s: coercivity, smoothness}, we show that the SA objective is coercive and smooth in the sub-level set of the cost. 
\Cref{s: gradient domination} establishes gradient domination with the explicit heterogeneity gap bound, and \Cref{s: proof of theorem} combines the two results into global convergence. 
% We then show global convergence under Assumption~\ref{asmp: heterogeneity} in \Cref{s: proof of theorem}. In \Cref{s: gradient domination}, we show that Assumption~\ref{asmp: heterogeneity} is indeed valid, providing upper bounds to $\het^{-1},\grad^{-1}$, by adapting the result of \citep{fujinami2025policygradientlqrdomain} to the lifted setting. 
Full proofs are deferred to the appendix.
% \drivelink. \footnote{The extended manuscript can be found here: \url{https://drive.google.com/file/d/1gPQn1IlEu_XFFg0EsAHVw--lGkxHKygm/view?usp=sharing}}

% In the remainder of this section, we outline the main steps to derive the extension. 
% the gap in the sample average cost between an arbitrary FIR policy $K$ and the optimal FIR policy $K^\star$ can be bounded by the norm of the gradient of the sample average cost evaluated at $K$, a property known as \emph{gradient domination} \citep{fazel2018global}. This suffices to ensure convergence of policy gradient to a minimizer of \eqref{eq: sa pg} under a suitable choice of stepsize.

% In particular, we begin by expressing the gradient of the sample average objective under FIR policies of order $H$ in closed-form by writing the dynamics of a lifted system consisting of an extended state with $H$ steps of memory. We then show why the result of \citep{fujinami2025domainrandomizationsampleefficient} does not immediately apply to this lifted system, and how to fix it. 
% \tas{highlight that FIR analysis is challenging because it is difficult to establish ``smoothness properties", coercivity etc. We should highlight this with words more in the paper}

\subsection{Coercivity and Smoothness of the Sample Average Objective}
\label{s: coercivity, smoothness}
% \tasb{I would prioritize bringing the discussion from lemma III.4 up to the end of this subsection sooner. }\tasb{I can't really follow the short proofs in this section}
We first show that the SA objective is coercive and $L$-smooth. This will in turn be used to ensure convergence of PG in \eqref{eq: sa pg} to a fixed point. Note that the LQR cost is known to be coercive and $L$-smooth over any sub-level set from \cite[Lemma~1]{hu2023toward}. However, it is not obvious if the same property holds for the lifted LQR objective. To this end, we establish the following lemma. 

\begin{lemma}[Coercivity and Smoothness% of sample average FIR-LQR Objective
]
    \label{lem: coercivity of sa FIR-LQR}
    For $i = 1, \cdots, M$, the FIR-LQR cost $J^i(K)$ is coercive over the set $\calK^i \triangleq \curly{K\in\R^{\du\times H\dx}:\rho(\bbA^i + \bbB^iK)<1}$, i.e. for any sequence $\{K^l\}_{l=1}^\infty\subset\calK^i$ we have $J^i(K^l)\rightarrow\infty$ if either $\norm{K^l}_2\rightarrow\infty$, or $K^l\rightarrow K$ with $\rho(\bbA^i + \bbB^iK)=1$. Therefore, the SA cost $\sa(K)$ is coercive on $\calK$ defined in \eqref{eq: calK-simultaneous-stable-set}. Furthermore, let $\calK_{\zeta} \triangleq \curly{K\in\calK : \sa(K)\leq \zeta}$ for $\zeta<\infty$ denote the sublevel set of the sample average objective. Then $\sa(K)$ is twice continuously differentiable and $L$-smooth on $\calK_{\zeta}$. 
\end{lemma}
% \tas{this is not a complete proof as it requires bounding the least singular value of Phi (which in turn depends on the norm of Phi since the diagonal is identity), which is not discussed. This is one of the technical challenges in this paper since we have to bound the norms of $\Phi$ and $K$ simultaneously. We should either remove the proof environment and discuss informally or do a proper full proof}
We sketch out the main idea of the proof, and defer the full details to Appendix \ref{s app: landscape of single fir-lqr}. To establish coercivity of single FIR-LQR cost, we need to show that $\norm{K}$ is bounded by $J^i(K)$. Since the noise in the lifted state space is degenerate, the techniques of~\cite{fujinami2025policygradientlqrdomain,fazel2018global} cannot be directly applied. Note that $x_t$ can be represented as a linear combination of the past noise $w_{t-l}, \forall l \leq t$ and some closed-loop state impulse response $\{\Phi_l(A^i,B^i,K)\}_{l\geq0}$
\[
    x_t = \sum_{l=0}^{t-1}\Phi_l(A^i,B^i,K)w_{t-1-l}.
\]
where $\Phi_l$ is a function of $A^i, B^i, K$. 
Since $u_t$ is a linear combination of a controller gain $K$ and past $H$ states from $x_{t-H+1}$ to $x_t$, it follows from $R=I$ and $W=I$ that
\begin{align}
    J^i(K) \geq \lim_{T\rightarrow\infty}\E^K[u_t^\top Ru_t] \geq \trace(K\Phi K^\top) \label{eq: Phi bound}
\end{align}
where $\Phi$ denotes that matrix consisting of $\{\Phi_l\}_{l=0}^{H-1}$. By exploiting the special structure of the state impulse responses, we can establish nonsingularity of $\Phi$ and boundedness of $K$ simultaneously. We show that \eqref{eq: Phi bound}, in turn, implies that 
\begin{align}
    \norm{K}_F \leq (1+\sqrt{J})^H - 1. \label{eq: K bound}
\end{align}
% where $f_K(J^i(K),H,\norm{A^i},\norm{B^i})$ 
% \tas{you are already using $f$ for  $\varepsilon_{het}$ maybe differentiate notation} 
% is an increasing function of $J^i(K)$, $H$, $\norm{A}$ and $\norm{B}$. 
Also for any $K\in\calK^i$, $J^i(K)$ goes to infinity when $K$ approaches the boundary of stability region. Therefore, $J^i(K)$ is coercive. 
Then, since $M\sa(K) \geq J^i(K)$, $\sa(K)$ is also coercive on $\calK$. 
Moreover, from \cite[Lemma~1]{hu2023toward}, $J^i(K)$ is twice continuously differentiable $C^2$ and thus $\sa(K)$ is also $C^2$. Therefore, from \cite[Thm.~1]{hu2023toward}, $\calK_{\zeta}$ is compact and thus $\norm{\nabla_K^2 \sa(K)}$ is bounded on $\calK_{\zeta}$. Let this uniform upper bound be $L$. Then $\sa(K)$ is $L$-smooth and satisfy the following inequality by the mean value theorem
\[
    \sa(K') - \sa(K) \leq \langle (K'-K), \nabla \sa(K)\rangle + \frac{L}{2}\norm{K'-K}_F^2
\]
for any $K, K'\in \calK_{\zeta}$. Note that $L$ is determined as
\begin{align*}
    L(\zeta) \coloneqq \max_{K\in\calK_\zeta}\norm{\nabla^2\sa(K)} \le \max_{i\in\brac{M}}\max_{K\in\calK_\zeta}\norm{\nabla^2J^i(K)}
    % L \leq f_L(&\sup_{i\in\{1, \cdots, M\}}J^i(K),\sup_{i\in\{1, \cdots, M\}}\norm{\nabla_KJ^i(K)}, H, \\
    % &\sup_{i\in\{1, \cdots, M\}}\norm{A^i},\sup_{i\in\{1, \cdots, M\}}\norm{B^i})
\end{align*}
which depends on $\zeta$ and the problem parameters only. 
As a consequence of \Cref{lem: coercivity of sa FIR-LQR}, \cite[Thm.~1]{hu2023toward} ensures that PG converges to a fixed point of $\sa(K)$. 
\begin{lemma}[Convergence to a Fixed Point]
    \label{lem: convergence to a fixed point}
    Let $\zeta_0 = J_{SA}(K^{(0)})$ and $\sa(K)$ be $L$-smooth on $\mathcal{K}_{\zeta_0}$.
    Consider the policy gradient method \eqref{eq: sa pg}.  Then, for any $0 < \alpha < \frac{2}{L}$, we have $K^{(n)}\in\mathcal{K}_{\zeta_0}$ and $J_{SA}(K^{(n+1)}) \leq J_{SA}(K^{(n)})$ for all n.
    Furthermore, we have the convergence of the gradient $\nabla_K J_{SA}(K) \rightarrow 0$ with the rate
    \begin{align}
        \min_{0\leq l\leq k}\norm{\nabla_K J_{SA}(K^{(l)})}^2_F \leq \frac{\zeta_0}{C(k+1)}, \label{eq: grad convergence rate}
    \end{align}
    with $C = \alpha - \frac{L\alpha^2}{2} > 0$.
\end{lemma}
Given that $K$ never leaves $\calK_{\zeta_0}$ under PG, we suppose $K\in\calK_{\zeta_0}$ from now on. This makes the rest of proof simpler because we can derive the uniform bounds on each sample's cost $J^i(K)$ on $\calK_{\zeta_0}$ that are independent of $K$. 
\begin{lemma}[Scenario Boundedness]
    \label{lem: scenario boundedness}
    Suppose Assumptions \ref{asmp: heterogeneity}, \ref{asmp: simultaneous stabilization} hold. For every $K\in\calK_{\zeta_0}$ and $i\in\brac{M}$, it holds that
    \[
        J^i(K) \le 2\sa(K) \le 32\Jbs
    \]
\end{lemma}
\begin{proof}
    Refer to Appendix \ref{s app: perturbation argument on the fir-lqr cost}. 
\end{proof}

In the next section, we establish that the SA objective satisfies gradient domination under small heterogeneity. This, in combination with coercivity and smoothness, can show global convergence. 
% that under Assumption \ref{asmp: heterogeneity}, global convergence can be achieved with gradient domination. 

\subsection{Gradient Domination of the Sample Average Objective}
\label{s: gradient domination}
% \tas{you should highlight this lemma and the discussion afterwards, bring those earlier}
Gradient domination of the SA objective only holds when the heterogeneity gap is small. We can formalize this by applying the Lyapunov perturbation arguments developed in \cite{fujinami2025policygradientlqrdomain, simchowitz2020naive}. This requires that the LQR instances have a nondegenerate process noise, and a non-degenerate state cost. However, the lifted LQR parameterization no longer satisfies this property, as $\lambda_{\min}(\bbW) = \lambda_{\min}(\bbQ) = 0$.
To address this issue, we derive new expressions for the LQR cost. 
% Consider an arbitrary lifted LQR with $H$ steps of memory defined by $\bbA$ and $\bbB$ extended form $A$ and $B$. 
Consider an arbitrary FIR policy of order $H$ and define
\begin{align}
    \label{eq: sigma H}
    \Sigma_H^i(K) &\coloneqq \bbW + \cdots + (\bbA_K^i)^{H-1}\bbW((\bbA_K^i)^{H-1})^\top \\
    \bar{\bbQ} &\coloneqq \diag\paren{\frac{Q}{H}, \frac{Q}{H}, \cdots, \frac{Q}{H}} \label{eq: bar Q}
\end{align}
We immediately see that there is a uniform lower bound on $\lambda_{\min}(\bar \bbQ)$ while the LQR cost remains the same if we replace $\bbQ$ with $\bar \bbQ$ in \eqref{eq: FIR cost}. Additionally, by the controllability of the lifted system, we can lower bound the minimum eigenvalue of $\Sigma_H^i(K)$. Combined with the uniform cost bound in \Cref{lem: scenario boundedness}, it holds that
\begin{align}
    \bar{c} \coloneqq \frac{1}{\paren{\sum_{h=0}^{H-1}(32H\Jbs)^{h/2}}^2} \le \lmin(\Sigma_H^i(K))
\end{align}
which is shown in \Cref{lem app: cost bounds on the FIR gain}. 
% This minimum eigenvalue is not uniform over the controllers $K$; however, in the \drivelink, we show that there exists nonzero $c$ such that $c\leq\lambda_{\min}(\Sigma_H^i)$ for some 
% \begin{align*}
%     c \triangleq \frac{1}{f_c(J^i(K),H,\norm{A^i},\norm{B^i})}.
% \end{align*} 
% where $f_c$ is an increasing function of $J^i(K), H, \norm{A^i}$ and $\norm{B^i}$. 
% where $f$ is the increasing function of $J^i(K)$ and $H$. 
% This implies that $\lambda_{\min}(\Sigma_H) \to 0$ only if $J(K) \to \infty$.
% We use these eigenvalue lower bounds to upper bound various system quantities in terms of the objective value.
% \tas{$f$ is used twice already for the other definitions, use a different notation}
With \eqref{eq: sigma H} and \eqref{eq: bar Q}, we reformulate the FIR-LQR objective.
\begin{lemma}[Reformulated FIR-LQR Objective]
    \label{lem: reparameterized lifted lqr}
    The lifted FIR-LQR cost can be written as
    \begin{align}
        J^i(K) = \trace\paren{(\bar \bbQ+K^\top RK){\Sigma}_K^i} = \trace\paren{\bar{P}_K^i\Sigma_H^i(K)} \label{eq: reparameterized lifted lqr obj}
    \end{align}
    where $\Sigma_K^i$ and $\bar{P}_K^i$ are the solution to the following Lyapunov equations
    % \tas{Remove bar from $\bar\Sigma$ throughout the paper it is equal to $\Sigma$. This creates unecessary confusion}
    \begin{align*}
         &{\Sigma}_K^i \coloneqq \bbA_{K}^H{\Sigma}_K^i(\bbA_{K}^H)^\top + \Sigma_H^i(K) \\
         &\bar{P}_K^i \coloneqq ((\bbA_K^i)^{H})^\top\bar{P}_K^i(\bbA_K^i)^{H} + \bar{\bbQ} + K^\top RK
    \end{align*}
\end{lemma}
\begin{proof}
    It can be confirmed by writing down the original state covariance matrix \eqref{eq: fir noise and Q} in the infinite sum form, and rearranging in the $H$-step propagation. Refer to \Cref{lem app: strucure of lifted problem} for further details.  
\end{proof}

% The sample average FIR-LQR cost always satisfies the following weaker version of gradient domination, even without any heterogeneity condition. 
Now, we discuss gradient domination of $\sa$. Without any heterogeneity condition, only the weaker condition holds. 
\begin{lemma}[Approximate gradient domination]
    \label{lem: approximate gradient domination}
    Let $K\in\calK$ and $\lambda^\star\coloneqq\min_i\lambda_{min}(\Sigma^i_{K^\star})>0$. Then the SA objective is bounded as
    \begin{align*}
    &\sa(K) - \sa(K^\star) \leq \frac{1}{2\lambda^\star}\paren{\norm{\nabla_K \sa(K)}^2_F + 4\norm{\sfR(K,K^\star)}^2_F}
\end{align*}
where the remainder term $\sfR(K,K^\star)$ is
\begin{align}
    \sfR(K,K^\star) \coloneqq \frac{1}{M}\sum_{i=1}^ME_K^i(\Sigma_{K^\star}^i - \Sigma_K^i).\label{eq: residual term}
\end{align}
\end{lemma}
\begin{proof}
    This result follows by applying \cite[Lemma III.3]{fujinami2025policygradientlqrdomain} to the lifted LQR instance along with \Cref{lem: scenario boundedness}. 
\end{proof}

To get exact gradient domination, we establish the upper bound of the remainder term under small heterogeneity. Note that $\sfR(K,K^\star)$ can be rearranged as
\[
    \sfR(K,K^\star) = \frac{1}{2M}\sum_i\nabla J^i(K)(\Sigma_K^i)^{-1}(\Sigma_{K^\star}^i - \Sigma_K^i)
\]
by $\nabla J^i(K) = 2E_K^i\Sigma_K^i$. Furthermore, the gradient of each cost can be decomposed into two parts: 
\[
    \nabla J^i(K) = \nabla \sa(K) + \frac{1}{M}\sum_{j\neq i}(\nabla J^i(K) - \nabla J^j(K))
\]
Thus, $\norm{\sfR(K,K^\star)}_F$ is made small as long as the norms of the SA gradient $\norm{\nabla \sa(K)}$ and the perturbation $\norm{\Sigma_{K^\star}^i - \Sigma_K^i}, \norm{\nabla J^i(K) - \nabla J^j(K)}$ are small. 
Given the continuity of Lyapunov equation on the stabilizing set, the perturbation term can be controlled by the heterogeneity bound $\barhet$. Also, let $\grad(K)\coloneqq\norm{\nabla\sa(K)}_F$ denote this gradient size. Then, the remainder term becomes small as long as $\grad(K)$ is also small, i.e. PG converges around the stationary point. We formalize these observations into the following lemma. 

% by \tas{typo: $\|\sfR \dots\|^2\le \dots$} $\sfR(K,K^\star)\leq \frac{1}{4} (J_{SA}(K) - J_{SA}(K^\star))$ under sufficiently small $\het,\grad$. 
% This in turn enables us to show gradient domination. 

% \citep{fujinami2025domainrandomizationsampleefficient} show this fact in the setting of static policies by applying Lyapunov perturbation arguments from \citep{fazel2018global, simchowitz2020naive} to the difference in state covariance matrices appearing in the remainder term. The application of these perturbation arguments requires that 
% Compared to \eqref{eq: fir noise and Q}, the state covariance matrix $\Sigma_K^i$ remains the same while $P_K^i \neq \bar{P}_K^i$ for $H\neq1$. 
% Now that we have the LQR objective with non-degenerate process noise and state penalty, we can instantiate the perturbation argument from \citep{fujinami2025policygradientlqrdomain} on the remainder term to get the following lemma. 
\begin{lemma}[Gradient Domination]
    \label{lem: gradient domination}
    Suppose Assumptions \ref{asmp: heterogeneity}, \ref{asmp: simultaneous stabilization} hold. Let $C_\gamma$ be an increasing function of $H, \tau_B, \Jbs$ and $\fgrad \coloneqq 8C_\gamma/\cb$. 
    Furthermore, suppose that the controller satisfies the gradient condition $K\in S$, where
    \begin{align}
        S \coloneqq \curly{K\in\calK_{\zeta_0}~:~\norm{\nabla \sa(K)}_F \le \frac{1}{\fgrad}}. 
    \end{align}
    Then gradient domination holds
    \[
        \sa(K) - \sa(K^\star) \le \frac{1}{\cb}\norm{\nabla \sa(K)}_F^2.
    \]
    \begin{proof}
        Refer to Appendix \ref{s app: gradient domination and global convergence}. 
    \end{proof}
    
    % Suppose $\sa(K) \leq 8\inf_K\sa(K)$. 
    % Then there exist  
    % % \tas{inconsistent definition compared to Assumption 1}
    % $\fhet, \fgrad > 0$ such that Assumption \ref{asmp: heterogeneity} is satisfied with
    % \begin{align*}
    %     \het^{-1}\leq\fhet, \quad \grad^{-1}\leq\fgrad
    % \end{align*}
    % where $\fhet, \fgrad$ are the functions of $\sup_{i\in\{1, \cdots, M\}}J^i(K)$, $\sup_{i\in\{1, \cdots, M\}}\norm{A^i}$,$\sup_{i\in\{1, \cdots, M\}}\norm{B^i}$, and $H$. 
    
    % as long as 
    % \[
    %     \norm{\brac{A(\theta_i) ~ B(\theta_i)} - \brac{A(\theta_j) ~ B(\theta_j)}} \leq \het
    % \]
    % for all $i, j \in \curly{1, \cdots, M}$, the remainder term can be bounded
    % \begin{align*}
    %     &\norm{\sfR(K,K^\star)}_F \leq \frac{1}{4}\norm{K - K^\star}_F
    % \end{align*}
    % for all $K\in S$, where 
    % % \tas{$\zeta_0$ should be defined more clearly it is hidden in the statement of ThIII.1, also how is $K-0$ relevant?}\tas{clash of notation with $S$ in the definition for the set of thetas in Assumption 1}
    % \[
    %     S = \curly{K \in \calK_{\zeta_0}:\norm{\nabla_K\sa(K)}_F \leq \grad}.
    % \]
    % % \tas{the way this is written, it reads like $\grad$ depends on the choice of $K$, shouldnt' it be independent of choice of $K$?, do you mean minimum over several $K$?}. 
    % Furthermore for all $K\in S$, the gradient domination holds, i.e.
    % \[
    % \sa(K) - \sa(K^\star) \leq 4\norm{\nabla_K\sa(K)}_F^2
    % \]
\end{lemma}
% \tas{I am worried that these functions depend on $K$ while in Assumption III.1 they are not supposed to. How to remove the sup over J(K)? I guess we should use the property of $J_{SA}(K)\le 8\dots$}
% \tas{We should properly define these functions with equation numbers in the appendix. We should have functions that do not depend on $K$. }

% To prove this result, we first apply the perturbation argument from \cite[Lemma~III.5-7]{fujinami2025policygradientlqrdomain} to the newly defined ${\Sigma}_K$ and $\bar{P}_K$ in \Cref{lem: reparameterized lifted lqr}, and show that the remainder term \eqref{eq: residual term} is upper bounded
% \begin{align}
%     \norm{\sfR(K,K^\star)}_F^2 \leq \frac{1}{4}(\sa(K)-\sa(K^\star)) \label{eq: remainder term ub}
% \end{align}
% as long as $\het^{-1}\leq\fhet, \quad \grad^{-1}\leq\fgrad$. Then, combining with \Cref{lem: approximate gradient domination} yields the gradient domination \eqref{eq: gradient domination}.
% The full proof can be found in the \drivelink. 

\begin{remark}[Limitation of Assumption \ref{asmp: heterogeneity}]
    Since $\fhet$ proposed in \Cref{lem: gradient domination} scales with the FIR order $H$, it does not fully capture the phenomenon observed in Example \ref{ex: a=1.1,-1.1}, which explains that increasing $H$ helps the controller stabilize the multiple systems with larger heterogeneity gap. 
    % Future work will focus on figuring out if the dependence of $\het$ on $H$ is fundamental. 
    However, the convergence result is interesting because even in the small heterogeneity gap regime, the optimization landscape induced by FIR policies remains highly nontrivial. Our result shows that policy gradient can still reliably find the optimum despite the nontrivial optimization landscape. 
    % See \Cref{sec: discussion} for further discussion on possible relaxations of the heterogeneity assumption. 
\end{remark}
% \begin{remark}[Limitation of Assumption \ref{asmp: heterogeneity}]
%     \tas{we could discuss this after the presentation of the result. Some extra context is needed in order to understand how limiting this assumption is. For example, does it have a poor dependence on $H$?}The heterogeneity condition presented in Assumption~\ref{asmp: heterogeneity} is a very strong limitation on the distribution of sampled systems. In particular, under Assumption~\ref{asmp: heterogeneity}, an LQR controller synthesized on any sample from the distribution will stabilize any other sample. Future work will focus on relaxing this assumption. However, the convergence result is interesting despite the strong heterogeneity condition because even in this regime, the optimization landscape induced by FIR policies remains highly nontrivial. Our result shows that policy gradient can still reliably find the optimum despite the nontrivial optimization landscape. See \Cref{sec: discussion} for further discussion on possible relaxations of the heterogeneity assumption. 
% \end{remark}
% Note that from \Cref{lem: coercivity of sa-LQR}, policy gradient is guaranteed to converge to the fixed point, and thus satisfying $S$ is also guaranteed. Based on all the lemmas introduced in this section, we prove the main theorem 

\subsection{Proof of \Cref{thm: gradient dominance of sa-lqr}}
\label{s: proof of theorem}
Finally, combining the landscape analysis of the SA objective in Section \ref{s: coercivity, smoothness} with gradient domination in Section \ref{s: gradient domination}, we prove \Cref{thm: gradient dominance of sa-lqr}. 
% \begin{proof}
% % With these bounds in hand, we can apply perturbation arguments to the remainder term $R(K, K^\star)$ to bound $\norm{\Sigma_{K^\star} - \Sigma_K}$ in terms of a small constant multiplied by $\norm{K-K^\star}$. As long as the initial iterate has a sufficiently small cost, and the heterogeneity condition is satisfied, this upper bound can be substituted into \Cref{lem: approximate gradient domination} to achieve gradient domination, as shown in \app. 
% From \Cref{lem: gradient domination}, $\norm{\nabla_K\sa(K)}_F \leq 1/\fgrad$ needs to be satisfied for gradient domination to hold. Thus from \eqref{eq: grad convergence rate} of \Cref{lem: convergence to a fixed point}, it suffices that
% $N \geq 2L\zeta_0\fgrad^2
% $
% given $\alpha = 1/L$ and $C = 1/2L$. Under this condition, global convergence is achieved from \Cref{lem: gradient domination}, when
% \[
%     \sa(K^{(N)}) - \sa(K^\star) \leq \cb^{-1}\norm{\nabla_K\sa(K)}_F^2 \leq \epsilon
% \]
% which holds after
% $N \geq 2L\zeta_0/\cb\epsilon$
% iterations. This establishes the lower bound on $N$.
% \end{proof}

\begin{proof}
    From \Cref{lem: convergence to a fixed point} and $\alpha = 1/L$, $\sa(K^{(n+1)})\le\sa(K^{(n)})$, and
    $\min_{0\le n\le N}\|\nabla\sa(K^{(n)})\|_F^2 \le 2L\zeta_0/(N+1)$.
    Let $l^\star$ attain this minimum. When $N\ge2L\zeta_0\fgrad^2$, it follows that $\|\nabla\sa(K^{(l^\star)})\|_F \le 1/\fgrad$.
    As Assumptions \ref{asmp-app: simultaneous stabilization}, \ref{asmp-app: heterogeneity assumption} hold and $K^{l^\star}\in\calK_{\zeta_0}$, \Cref{lem: gradient domination} applies at $K^{(l^\star)}$:
    \[
        \sa(K^{l^\star}) - \sa(K^\star) \le \frac{1}{\cb}\norm{\sa(K^{(l^\star)})}_F^2 \le \frac{2L\zeta_0}{N+1}\le\epsilon
    \]
    after $N\ge2L\zeta_0/\cb\epsilon$ iterations. 
    Finally from monotonicity, $\sa(K^{(N)}) - \sa(K^\star) \le \sa(K^{l^\star}) - \sa(K^\star) \le \epsilon$. 
\end{proof}

\section{Curriculum Learning-based Policy Gradient}
\label{s: curriculum learning}
In Section \ref{sec: convergence}, we show that PG converges for any fixed FIR order $H$. The choice of $H$ is a design parameter: smaller $H$ is preferable because it requires less memory and online computation, while larger $H$ can improve simultaneous stabilization, as demonstrated in Example~\ref{ex: a=1.1,-1.1}. We propose a curriculum learning-based algorithm that gradually increases $H$ until satisfactory performance is achieved, whenever such performance is achievable with FIR policies. \Cref{alg:curriculum_learning} consists of two components: discount annealing, which finds an initial stabilizing controller for a given $H$, and FIR order selection, which gradually increases $H$ until a satisfactory policy is found.  

\subsection{Discount Annealing}
Generally, it is hard to find a controller that simultaneously stabilizes a collection of systems. 
\cite[Sect.~IV]{fujinami2025policygradientlqrdomain} presents a scheme that consists of successively optimizing discounted versions of the sample average objective in which the stage cost of the LQR problems at time $t$ are multiplied by the discount factor $\gamma^t$ or, equivalently, the state and input matrices for each sample are multiplied by $\sqrt{\gamma}$ with $0 < \gamma \leq 1$. For any controller $K$, if $\gamma < \min_{i \in \curly{1,\dots,M}} \rho(\bbA(\theta_i)+\bbB(\theta_i)K)^{-2}$, then $K$ stabilizes $\sqrt{\gamma} \bbA(\theta_i)$ and $\sqrt{\gamma} \bbB(\theta_i)$ for $i=1,\dots,M$.\footnote{Scaling $\bbA$ and $\bbB$ by $\sqrt{\gamma} < \rho(\bbA + \bbB K)^{-1}$ is equivalent to introducing a discount factor in the LQR cost which ensures the cost is finite even for a controller which does not stabilize $\bbA$, $\bbB$.} 
Then $\gamma$ is chosen to gradually increase from $0$ to $1$ such that the discounted versions of the problem always satisfy the condition on the initial iterate. \citep{fujinami2025policygradientlqrdomain} shows that such a scheme is sufficient to ensure convergence to the optimal solution of the SA objective in the setting of static policies starting from an arbitrary initial iterate as long as the heterogeneity assumption holds. However, the argument is not specific to static policies, and immediately extends to the setting of FIR policies.

% This observation suggests that if we gradually increase the discount factor until it recovers the original system, we can obtain the simultaneously stabilizing controller as long as this discount annealing scheme converges. 

%  A related scheme is discussed in \Cref{s: curriculum learning}. 

\subsection{FIR Order Selection}
Given a collection of systems, the FIR order $H$ required to simultaneously stabilize all the systems is undecidable when $M\ge3$ \citep{undecidableBlondel1993}. However, \citep{Blondel1993simultaneous} suggests that if systems are simultaneously stabilizable, there exists a proper controller with sufficiently rich complexity that achieves stabilization. This motivates a curriculum learning strategy in which we run policy gradient while gradually increasing $H$. In this way, we can efficiently find a small enough $H$ that provides satisfactory performance. In \Cref{sec: numerical}, we demonstrate the effectiveness of this strategy that progressively increases $H$. 

% \tas{the class of FIR controllers does not include all possible proper controllers, so this is not fully accurate. }. From this, we employ curriculum learning which we run the policy gradient by gradually increasing $H$ until we find some controller which stabilizes all instances and provides satisfactory performance.\tas{what about increasing H just for performance, not stabilization?}
%The problem of finding the smallest $H$ for the given samples is left for the future work. 

\begin{algorithm}[tbp]
    \caption{Curriculum Learning-Based Policy Gradient}
    \label{alg:curriculum_learning}
    \begin{algorithmic}[1]
    \State \textbf{Input:} collection of systems $\theta_1, \dots, \theta_M$, optimization tolerance $\epsilon$, initial cost threshold $C_{\textnormal{\texttt{init}}}$, update multiplier $C_{\gamma_1}$, $C_{\gamma_2}$, update threshold $\tau_\gamma$
    \State Initialize FIR order $H \gets 1$, ~ $K\gets0$
    \State Find largest $\gamma \in (0,1]$ s.t.
    $J_{SA}(K \mid \gamma, H) \le C_{\textnormal{\texttt{init}}}$\label{alg: init gamma}
    \While{$\gamma<1$}
        \State Apply PG \eqref{eq: sa pg} to find $K'$ such that
        \[
            J_{SA}(K' \mid \gamma, H) - \inf_{\tilde K} J_{SA}(\tilde K \mid \gamma, H) \leq \epsilon
        \] \label{alg: K update}
        \State $K \gets K'$        
        \State Find a discount factor $\gamma' \in [\gamma,1]$ such that
        \begin{align*}
            C_{\gamma_1} J_{SA}(K \mid \gamma, H)
            &\;\leq\;
            J_{SA}(K \mid \gamma', H) \\
            &\;\leq\; C_{\gamma_2} J_{SA}(K \mid \gamma, H)
        \end{align*} \label{alg: gamma update}
        % \If{$\gamma'$ such that \eqref{eq:gamma_search} holds does not exist 
        %     \State $H \gets H + 1$
        %     \State \textbf{break}
        % \EndIf
        % \bruce{The choice of when to update requires some more thought/explanationmy suggestion is below. We define some function $f$.}
        % \If{$\gamma' < f(\gamma)$ or $\gamma=1$ and sample cost is unsatisfactory}
        %     \State $H \gets H+1$
        % \EndIf
        \If{$\Delta_\gamma \coloneqq \gamma' - \gamma \leq \tau_\gamma$}
            \State $H \gets H+1$, ~ $K\gets0$ \label{alg: H update}
            \State \textbf{continue}
        \EndIf
        \State $\gamma \gets \gamma'$ 
    \EndWhile
    \State Run \eqref{eq: sa pg} initialized at $K$ to find $K'$ such that
    \[
        J_{SA}(K') - \inf_{\tilde K} J_{SA}(\tilde K) \leq \epsilon
    \]
    \State \textbf{return} $K'$
    \end{algorithmic}
\end{algorithm}

\subsection{Curriculum Learning-Based Policy Gradient}
\Cref{alg:curriculum_learning} is the pseudocode used in numerical experiments. We first choose the largest $\gamma$ such that $K=0$ can simultaneously stabilize a collection of systems, and such that the discounted sample average FIR-LQR cost given the FIR order $H$ $\sa(K\mid\gamma,H)$ is bounded by $C_{\textnormal{\texttt{init}}}$. Then we gradually update $K$ (line \ref{alg: K update}) and $\gamma$ (line \ref{alg: gamma update}) in a dual manner until we reach $\gamma = 1$. Note that $C_{\textnormal{\texttt{init}}}$, $C_{\gamma_1}$ and $C_{\gamma_2}$ need to be  chosen carefully: $C_{\textnormal{\texttt{init}}}$ and $C_{\gamma_2}$ ensure the initial cost condition $\sa(K^{(0)})\leq8\sa(K^\star)$ is satisfied when $\gamma$ reaches $1$, while $C_{\gamma_1}$is used to obtain a finite iteration guarantee. See \cite[Algorithm~1]{fujinami2025policygradientlqrdomain} for a concrete example. To determine when to update $H$, we introduce the update threshold $\tau_\gamma$. When a collection of systems is not simultaneously stabilizable under a given $H$, there exists $\gamma_\star < 1$ such that
\[
    \gamma_\star = \sup_{\gamma,K}\curly{\gamma:\rho(\sqrt{\gamma}\bbA^i + \sqrt{\gamma}\bbB^iK)<1, \forall i\in \{1, \cdots, M\}}.
\]
Therefore in this case, even though it is always possible to find $\gamma'$ such that line \ref{alg: gamma update} holds because of local continuity of $\sa(K)$, the update $\Delta_\gamma \coloneqq \gamma' - \gamma$ becomes incremental as $\gamma$ gets close to $\gamma_\star$. Thus if $\Delta_\gamma\leq\tau_\gamma$ for some threshold, we increase $H$ to enlarge the controller parameterization (line \ref{alg: H update}), then initialize $\gamma$ again. Additional details are provided in the code repository.\footnote{Code available at \href{https://github.com/Tesshuuuu/PolicyGradient_with_DomainRandomization}{this GitHub repository}}
% \tas{define $J_{SA}(K \mid \gamma, H)$}
\section{Numerical Experiments}
\label{sec: numerical}
\subsection{Policy Gradient over FIR Controller classes}
We conduct numerical experiments on the discretized and linearized inverted pendulum defined by
\begin{align}
    \label{eq: linearized pendulum}
    A = \bmat{1 & dt \\ \frac{g}{\ell} dt & 1}, 
    B= \bmat{0 \\ \frac{dt}{m \ell^2}}.
\end{align} 
We suppose that the parameters $dt = 0.01$ and $g=10$ are known. The unknown parameters $m$ and $\ell$ are modeled with a uniform distribution over the interval $[0.5, 5.0]$. Starting from the FIR order $H=1$, we gradually increase $H$ when the discount factor does not increase sufficiently toward 1. \Cref{fig: fir simulation} demonstrates the application of \Cref{alg:curriculum_learning} to $M=10$ systems sampled from the distribution of systems \eqref{eq: linearized pendulum}. We plot the convergence of the SA cost $\sa(K)$, the evolution of the controller gains $K$, and the updates of the discount factor $\gamma$, where phase boundaries indicate the updates of the FIR order $H$. As we can observe, the LQR cost starts to decrease substantially once the FIR order reaches $H=5$, and from that point onward it converges to the minimizer of $\sa(K)$. The plots of the controller gains and discount factor further show that, for smaller FIR orders, the controller gains start to grow rapidly in magnitude as $\gamma$ approaches $1$. This suggests that the sampled systems are difficult to stabilize with a lower order FIR controller. However, as $H$ increases, the controller can depend on longer history, and eventually this additional memory enables simultaneous stabilization, after which PG converges.  
% as suggested in \Cref{thm: gradient dominance of sa-lqr} \tas{I would remove the ``as suggested in..." since we don't know if the conditions of the theorem are met or not for this example}. 

\begin{figure}
    \centering
    \includegraphics[width=0.85\linewidth]{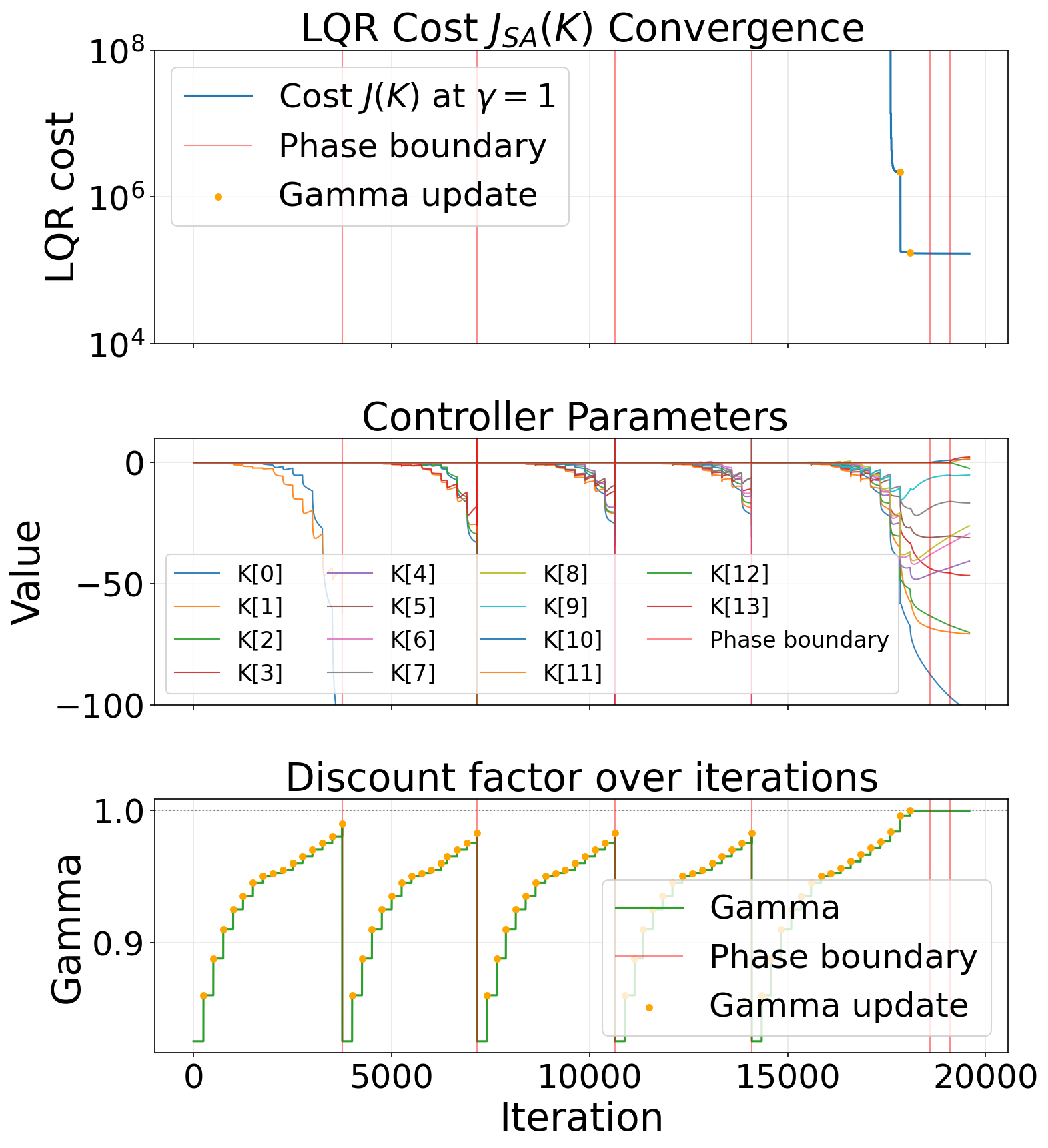}
    \caption{Convergence of policy gradient over FIR controller classes with domain randomization. A controller initialized at $K=0$ converges to the optimal controller via discount annealing and curriculum learning. Phase boundary 
    % \tas{explain term phase boundary in the main text too} 
    represents the increase of the FIR order $H$. The LQR cost starts to converge after $H = 5$. }
    % \tas{Increase font size!} }
    % \tas{increase font sizes in figures, make font sizes consistent across figures}}
    \label{fig: fir simulation}
\end{figure}

\subsection{Alternative Dynamic Parameterization}
\label{s: alternative dynamic}
We next consider an alternative dynamic parameterization. 
The general dynamic policy can be parameterized as follows 
\begin{equation}
    \begin{split}
        \bar\xi_{t+1} &= \bar{A}_K\bar\xi_t + \bar{B}_Kx_t \\
        u_t  &= \bar{C}_K\bar\xi_t + \bar{D}_Kx_t
    \end{split} \label{eq: dynamic K}
\end{equation}
where $\bar\xi_t\in\R^{(H-1)\dx}$. 
Since the FIR controller class admits the following parameterization, FIR is one realization of \eqref{eq: dynamic K}. 
\begin{equation}
    \begin{split}
        A_K = \begin{bmatrix} 0 & 0 \\ I_{(H-2)\dx} & 0 \end{bmatrix}, \quad B_K = \begin{bmatrix} I_{\dx} \\ 0 \end{bmatrix} \\
        C_K = [K_1 ~ \cdots ~ K_{H-1}], \quad D_K = K_0  
    \end{split} \label{eq: fir-dynamic}
\end{equation}
Now we compare the performance of PG over several dynamic controller classes on the following example. 
\begin{example}
    \label{ex:scalar-lqr-unstable}
    Consider the scalar system
    \[
        x_{t+1} = ax_t + u_t + w_t, ~ a\in\{-2.0, 2.0\}
    \] 
\end{example}

\begin{figure}[t]
    \centering
    \begin{subfigure}{0.85\columnwidth}
        \centering
        \includegraphics[width=\columnwidth]{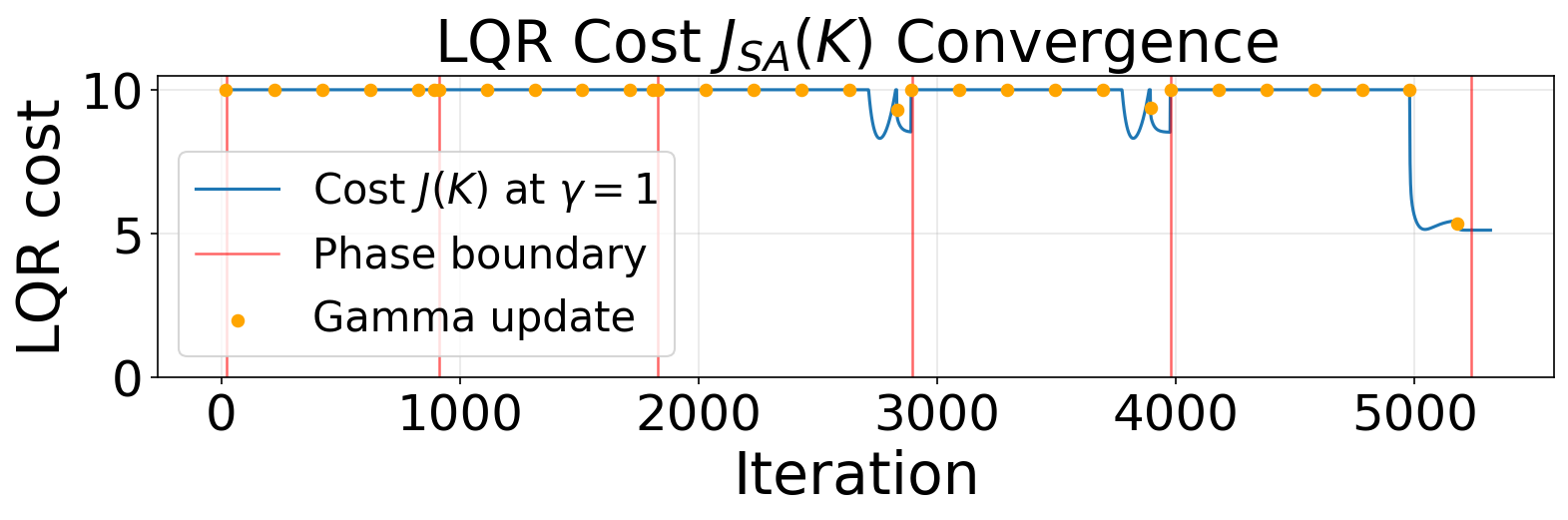}
        \caption{An FIR controller converges at $H = 5$.}
        \label{fig: fir}
    \end{subfigure}

    % \vspace{0.5em}

    \begin{subfigure}{0.85\columnwidth}
        \centering
        \includegraphics[width=\columnwidth]{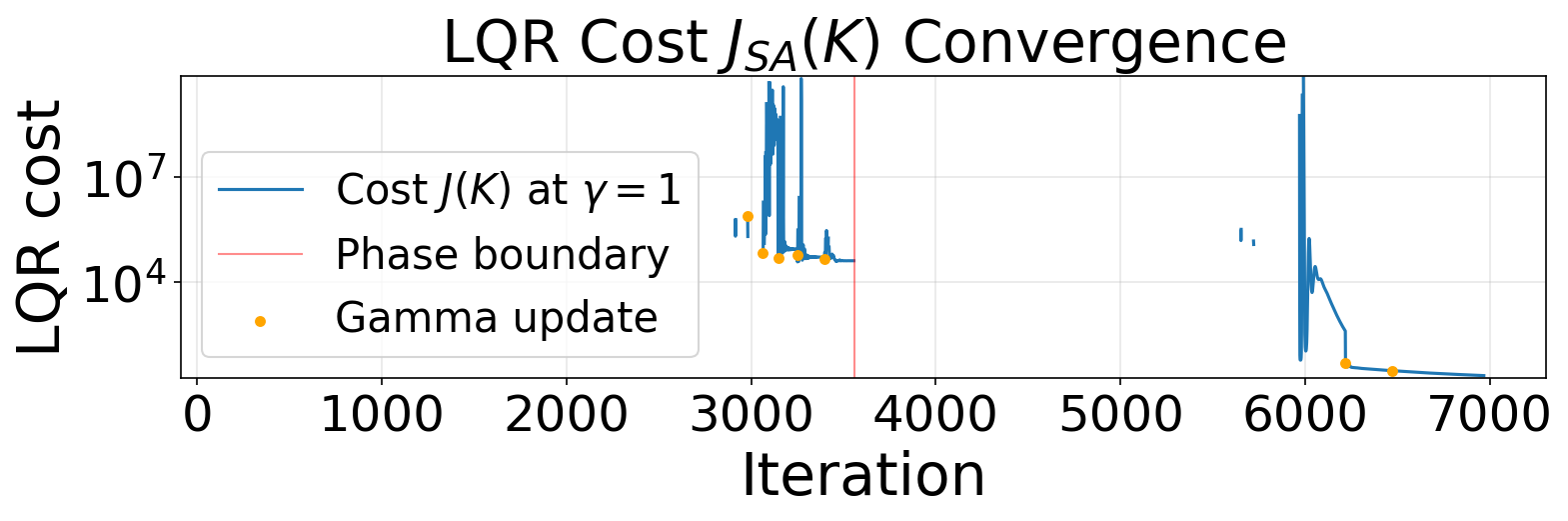}
        \caption{A general dynamic controller \eqref{eq: dynamic K} converges at $H=2$, but optimization is unstable.}
    \label{fig: dynamic}
    \end{subfigure}

    % \vspace{0.5em}

    % \begin{subfigure}{0.8\columnwidth}
    %     \centering
    %     \includegraphics[width=\columnwidth]{figures/Tutorial_canonical_scalar.png}
    %     \caption{Policy gradient on Example \ref{ex:scalar-lqr-unstable} with a canonical form controller \eqref{eq: canonical}. It converges at $H=2$.}
    % \label{fig: canonical}
    % \end{subfigure}

    \caption{Policy gradient on \Cref{ex:scalar-lqr-unstable} with FIR and general dynamic controllers.}
    \label{fig:threeplots_vertical}
\end{figure}

% \begin{figure}[tbp]
%     \centering
%     \includegraphics[width=0.85\linewidth]{figures/fir_lqr_pg_progressive_AB.png}
%     \caption{Policy gradient on \Cref{ex:scalar-lqr-unstable} with a FIR controller. It converges at $H = 5$.}
%     \label{fig: fir}
% \end{figure}
% \begin{figure}[tbp]
%     \centering
%     \includegraphics[width=0.85\linewidth]{figures/Tutorial_abcd_progressive_scalar.png}
%     \caption{Policy gradient on \Cref{ex:scalar-lqr-unstable} with a general dynamic controller \eqref{eq: dynamic K}. It converges at $H=2$ but optimization is unstable.}
%     \label{fig: dynamic}
% \end{figure}

As shown in Figures \ref{fig: fir} and \ref{fig: dynamic}, PG converges in both cases, but the FIR controller attains a lower cost in fewer iterations than the dynamic controller ($5.1$ vs. $21.8$). Although the FIR controller uses longer history ($H=5$) than the dynamic controller ($H=2$), curriculum learning makes optimization stable. 
This highlights the tradeoff between controller complexity and optimization ease.
% illustrated in \Cref{fig: opt-stabilize tradeoff} \tas{missing}. 
Although FIR controllers form a strict subclass of dynamic controllers and are therefore less expressive, their smaller number of degrees of freedom makes optimization easier and leads to more efficient optimization. Results for another dynamic parameterization can be found in Appendix \ref{s app: further experiments in dynamic form}. 

% Another useful parameterization is the canonical form controller
% \begin{equation}
%     \label{eq: canonical}
%     \begin{split}
%         &A_K = \begin{bmatrix} -d_3 & -d_2 & -d_1 & -d_0 \\ 1 & 0 & 0 & 0 \\ 0 & 1 & 0 & 0 \\ 0 & 0 & 1 & 0 \end{bmatrix}, \quad B_K = \begin{bmatrix} 1 \\ 0 \\ 0 \\ 0 \end{bmatrix}, \\
%         &C_K = \begin{bmatrix} n_3 & n_2 & n_1 & n_0 \end{bmatrix}, \quad D_K = 0
%     \end{split}
% \end{equation}
% % \begin{figure}[tbhp]
% %     \centering
% %     \includegraphics[width=0.85\linewidth]{figures/canonical_sclar.png}
% %     \caption{Policy gradient on Example \ref{ex:scalar-lqr-unstable} with a canonical form controller \eqref{eq: canonical}.}
% %     \label{fig: canonical}
% % \end{figure}
% As shown in \Cref{fig: canonical}, the canonical form is also effective for identifying simultaneously stabilizing controller. 
% Since it has more parameters than the FIR class, it only uses $H=2$ history while optimization remains stable. This further supports the tradeoff that limiting the degree of freedom of the controller class can simplify policy optimization while still achieving strong performance. However, compared with FIR controller in \Cref{fig: fir}, the canonical class attains a suboptimal solution. One potential explanation is that, by limiting $D_K = 0$,  the canonical form may fail to represent the best controller.
% \tas{any advantage compared to canonical?}

\section{Discussion}
\label{sec: discussion}
There are several possibilities for extensions of the results presented in this paper.

\begin{itemize}[noitemsep, nolistsep]
\item \textbf{Nonlinear parameterizations}
Although adding memory to the controller appears to improve simultaneous stabilizability, there still exist some examples that no linear controller can stabilize, such as \Cref{ex:scalar-lqg-unstable}. The use of nonlinear parameterization to overcome this problem is a promising direction. 
\begin{example}
    \label{ex:scalar-lqg-unstable}
    Consider the scalar system
    \[
        x_{t+1} = x_t + bu_t + w_t, ~ b\in\{-1.1, 1.1\}
    \] 
    There does not exist a static controller nor a dynamic controller that can stabilize the system. 
\end{example}

\item \textbf{Relaxed heterogeneity assumptions: } 
The heterogeneity assumption $\het$ is introduced for the optimality analysis of the DR-LQR objective, but the sufficient condition in \Cref{lem: gradient domination} suggests that $\het$ scales inversely with the FIR order $H$. This appears to conflict with the stabilization perspective, in which more expressive controllers generally increase the possibility of simultaneous stabilization by shaping the closed-loop geometry more favorably. Investigating whether this intuition from simultaneous stabilization can be leveraged in the optimality analysis is an intriguing direction for future work. 
% Although we analyze the closeness condition from the viewpoint of open-loop heterogeneity, this perspective does not fully capture the similarity among the resulting closed-loop systems. This is why theory developed in this paper does not explain why increasing memory help simultaneously stabilize systems\tas{I would rephrase this since we know that more expressive controllers give more chances for simultaneous stabilization. I guess you mean that despite the strict heterogeneity assumptions for convergence we still have stability. But then again the heterogeneity assumption is to analyze optimality, not stabilizability so let's not conflate these two }. As alternative perspectives, a common Lyapunov function or a bisimulation function may offer the promising directions \tas{to use them for what?}.

\end{itemize}
\section{Conclusion}
Our work motivates the use of history-dependent policies for domain randomization by considering collections of systems for which no static linear feedback controller can simultaneously stabilize the collection, whereas a history-dependent policy parameterized by a lower order finite impulse response controller can. Motivated by this fact, we prove the convergence of policy gradient methods over finite impulse response policies applied to minimize the sample average domain randomized linear quadratic control objective. 
We believe that further study of optimal policy classes for domain randomization serves as a stronger theoretical foundation for  reinforcement learning in robotics.

\section*{Acknowledgements}
% We thank Leonardo Toso for instructive conversations. 
TF is supported by JASSO Exchange Support program. TF and GP are supported in part by NSF TRIPODS EnCORE 2217033. BL is supported by an ETH AI Center Postdoctoral Fellowship. NM is supported by NSF Award SLES-2331880.

\bibliographystyle{IEEEtran}
\bibliography{refs}

\onecolumn

\section{Appendix}

\useRomanappendicesfalse
\appendices
\renewcommand{\thesubsection}{\thesection.\arabic{subsection}}
\renewcommand{\thesubsectiondis}{\thesubsection}

\setcounter{equation}{0}
\renewcommand{\theequation}{A.\arabic{equation}}

% \etocsettocstyle{\subsection*{Contents}}{}
% \localtableofcontents

\section{Notation and preliminaries}
\label{s: notation}
We briefly recall the setting so that the appendix is self-contained. For stabilizable sample parameters $\theta_i = (A^i,B^i), i\in\brac{M}$, drawn from $p_\Theta$ supported on $\calS\subseteq\R^{d_\theta}$, the dynamics are $x_{t+1} = A^ix_t + B^iu_t + w_t$ with $w_t\sim\calN(0,I)$, and the per-system cost of a policy $K$ is $J^i(K) = \lim_{T\rightarrow\infty}\frac{1}{T}\E^K\brac{\sum_{t=0}^T x_t^TQx_t + u_t^TRu_t \middle|\theta_i}$, where we assume $Q\succeq I$ and $R=I$. We consider the FIR policy class of order $H$, 
\[
    u_t = \sum_{h=0}^{H-1}K_hx_{t-h}, \quad K = \begin{bmatrix}
        K_0 & K_1 & \cdots & K_{H-1}
    \end{bmatrix}\in\R^{\du\times H\dx}
\]
Policy gradient is run on the sample-average cost
\begin{align}
    \sa(K) \coloneqq \frac{1}{M}\sum_{i=1}^MJ^i(K), \quad K^{(n+1)} = K^{(n)} - \alpha\nabla_K\sa(K^{(n)}) \label{eq app: pg}
\end{align}
Stacking the state history $\xi_t = \begin{bmatrix}
x_t^\top & x_{t-1}^\top & \cdots & x_{t-H+1}^\top\end{bmatrix}^\top\in\R^{H\dx}$ yields the lifted system with static feedback $u_t = K\xi_t$:
\begin{align}
    \xi_{t+1} &= \bbA^i\xi_t + \bbB^i u_t + \bbB_w w_t, \quad \bbA^i = \begin{bmatrix}
        A^i & 0 & \cdots & 0 \\
        I & 0 & \cdots & 0 \\
        \vdots & \ddots &  & \vdots \\
        0 & \cdots & I & 0
    \end{bmatrix}, ~ 
    \bbB^i =  \begin{bmatrix}
        B^i \\ 0 \\ \vdots \\ 0
    \end{bmatrix}, ~ 
    \bbB_w =  \begin{bmatrix}
        I \\ 0 \\ \vdots \\ 0
    \end{bmatrix}. \label{eq app: lifted system}
\end{align}
We write $\bbA_K^i \coloneqq \bbA^i + \bbB^iK$, $\bbW \coloneqq \bbB_w\bbB_w^\top = \diag(I, 0, \cdots, 0)$, $\bbQ \coloneqq \diag(Q,0,\cdots,0)$. When $\rho(\bbA_K^i)<1$, the per-system cost takes the matrix form
\[
    \Sigma_K^i \coloneqq \dlyap(\bbA_K^{i\top}, \bbW) = \sum_{t\ge0}(\bbA_K^i)^t\bbW(\bbA_K^i)^{t\top}, \quad J^i(K) = \trace\paren{(\bbQ+K^\top RK)\Sigma_K^i}
\]
where $P=\dlyap(A,X)$ is the unique PSD solution to $P = A^\top PA+X$.

Throughout, $\norm{\cdot}$ is the spectral norm, $\norm{\cdot}_F$ the Frobenius norm, $\norm{\cdot}_{\{\cdot,F\}}$ either of the two, $\rho(\cdot)$ the spectral radius, and $\lmin(\cdot)$ the smallest eigenvalue. We use the following objects. 
\begin{itemize}
    % \item For a matrix A, let $\norm{A}$ denote the operator norm and $\norm{A}_F$ the Frobenius norm. When a statement holds for either norm, we write $\norm{A}_{\{\cdot,F\}}$.
    \item \textit{Stabilizing and sublevel sets.} $\calK^i \coloneqq \curly{K\colon\rho(\bbA_K^i)<1}$, ~ $\calK\coloneq\cap_{i\in\brac{M}}\calK^i$, ~ $\calK_\zeta\coloneqq\curly{K\in\calK\colon\sa(K)\le\zeta}$. 
    \item \textit{Heterogeneity and gradient size}. $\het \coloneqq \max_{\theta_i, \theta_j\in\calS}\curly{\norm{A^i - A^j}, \norm{B^i - B^j}}$, $\grad(K) \coloneqq \norm{\nabla_K \sa(K)}_F$. Note that $\grad(\cdot)$ is a quantity evaluated along the analysis. 
    \item \textit{Model size.}$\tau_B \coloneqq \max\curly{1, \max_{\theta\in\calS}\norm{B(\theta)}}$, ~ $\nu(K) \coloneqq 1 + \norm{K}$
    % \item $\zeta(K) \triangleq 1 + \norm{K}^2$
    % \item $\norm{X + Y}^2\leq2\norm{X}^2+2\norm{Y}^2$
    % \item $X^n - Y^n = \sum_{t=0}^{n-1}X^{n-1-t}(X-Y)Y^t$
    \item $H$-step objects. $\bbA_{K,H}^i \coloneqq (\bbA_K^i)^{H}$, ~ $\Sigma_H^i(K) \coloneqq \sum_{t=0}^{H-1}(\bbA_K^i)^t\bbW(\bbA_K^i)^{t\top}$, ~ $\bar{\bbQ}\coloneqq\diag\paren{\frac{1}{H}Q, \cdots, \frac{1}{H}Q}\succeq\frac{1}{H}I$
    % \item $\Sigma_K = \bbA_K\Sigma_K\bbA_K^\top + \bbW = \bbA_{K,H}{\Sigma}_K\bbA_{K,H}^\top + \Sigma_H$
    \item \textit{Value matrices}. For $K\in\calK^i$, 
    \[
        P_K^i \coloneqq \dlyap(\bbA_K^i, ~ \bar\bbQ + K^\top RK) \footnote{Note that $\bar\bbQ$ is used instead of $\bbQ$ in \eqref{eq: fir noise and Q} because the cost is same either way}, ~ \bar{P}_K^i \coloneqq \dlyap(\bbA_{K,H}^i, ~ \bar\bbQ + K^\top RK).
    \]
    \item \textit{Gradient}. For $K\in\calK^i$, 
    \[
        \nabla J^i(K) = 2E_K^i\Sigma_K^i, \quad E_K^i \coloneqq \paren{R + \bbB^{i\top}P_K^i\bbB^i}K + \bbB^{i\top}P_K^i\bbA^i
    \]
    \item \textit{Optimal policies}. $K_i^\star$ denotes the optimal FIR gain of system $i$, and $K(\theta)$ that of a parameter $\theta\in\calS$. $K^\star\in\argmin_{K\in\calK}\sa(K)$ and
    \begin{align}
        \label{eq app: sup optimal gain}
        \bar{J}_\star \coloneqq \sup_{\theta\in\calS}J^{\theta}(K(\theta)) < \infty
    \end{align}
    Since $Q\succeq I$, $R=I$, we freely use $\bar{J}_\star \ge \trace Q \ge 1$. 
    % $P_K = \bbA_K^\top P_K\bbA_K + \bar{\bbQ} + K^\top RK$
    % \item $\bar{P}_K = (\bbA_{K,H})^\top\bar{P}_K\bbA_{K,H} + \bar{\bbQ} + K^\top RK$ 
    % \item $\dlyap(X,Y)$ is the unique PSD solution to $X^\top PX - P = Y$
    % \item $\Psta, P$ are the solution to the discrete algebraic Ricatti equation (or $\dare$) of LQR and lifted LQR.
    % \item $\norm{\bbA_K^i} \leq \alpha \triangleq \poly((\sqrt{J^i(K)})^H, A^i,B^i)$ 
    % \item $c \triangleq \lambda_{min}(\Sigma_K^i) = \poly((\sqrt{J^i(K)})^H, A^i,B^i)$  \Tesshu{to be shown, not given}\,\bruce{Let's not define it here then}
    % \item $C \triangleq \sum_{j=0}^H \alpha^{2j} = \poly((\sqrt{J^i(K)})^H, A^i,B^i)1$
\end{itemize}

\section{Assumptions and proof roadmap}
\begin{assumption}[Initial policy]
    \label{asmp-app: simultaneous stabilization}
    An initial gain $K^{(0)}\in\calK$ is available, and $\zeta_0 \coloneqq \sa(K^{(0)}) \leq 8\min_{K\in\calK}\sa(K)$.
\end{assumption}
\begin{assumption}[Small heterogeneity]
    \label{asmp-app: heterogeneity assumption}
    $\het\le\barhet$, where $\barhet = 1/\fhet(H,\tau_B,\bar{J}_\star)$ is the explicit threshold defined in \Cref{lem app: gradient domination of sa}. 
    % Let $\calS\subseteq\R^{d_\theta}$ be such that $p_\Theta(\theta) = 0$ for $\theta\notin\calS$. There exist $\het>0$, such that the following heterogeneity condition holds:
    % \[
    %   \forall\theta_1, \theta_2\in\calS,\,   \norm{\brac{A(\theta_1) ~ B(\theta_1)} - \brac{A(\theta_2) ~ B(\theta_2)}} \leq \het,
    % \]
\end{assumption}
% \begin{assumption}[Bounded Sample Average Cost Assumption]
%     \label{asmp-app: bounded cost assumption}
%     \Tesshu{not correct}
%     We assume that $\forall K \in \calK$ satisfies
%     \[
%         \sa(K) \leq \zeta
%     \]
%     for $\zeta<\infty$. 
% \end{assumption}
% \subsection{Proof Roadmap}
\textbf{Roadmap}. The lifted system does not immediately follow the static LQR theory \cite{fujinami2025policygradientlqrdomain, fazel2018global, hu2023toward} since $\lmin(\bbW) = \lmin(\bbQ) = 0$ for $H\ge2$. We proceed in five steps. 
\begin{enumerate}
    \item \textbf{The landscape of a single FIR-LQR cost} (\Cref{s app: landscape of single fir-lqr}).
    Structural facts of the lifted problem remove the degeneracy (\Cref{lem app: strucure of lifted problem}). This yields coercivity, smoothness, and gradient domination of a single cost. 
    \item \textbf{Optimal FIR policy of a single system} (\Cref{s app: optimal fir policy}). 
    Memory is useless for a single system. The lifted $\dare$ solution admits a block-diagonal structure, and the optimal FIR gain is the static LQR gain padded with zeros (\Cref{lem app: policy evaluation with the optimal gain}). This yields $\dare$ norm bounds that are later used to compare multiple systems. 

    \item \textbf{Uniform bounds} (\Cref{s app: uniform bounds}). Under small heterogeneity, every FIR gain in the initial sub-level set ($K\in\calK_{\zeta_0}$) satisfies uniform cost bounds, which allow all the landscape constants to be uniformly bounded by the functions independent of $K$ (\Cref{cor:uniform}). 
    
    \item \textbf{Perturbation bounds} (\Cref{s app: perturbation} - \ref{s app: perturbation across gains}). We bound the perturbation on $\Sigma_K, P_K, E_K, \nabla J(K)$ across systems and gains. 
    
    \item \textbf{Gradient Domination of the sample average FIR-LQR cost} (\Cref{s app: gradient domination and global convergence}). Approximate gradient domination plus quadratic growth produces gradient domination of $\sa$ (\Cref{lem app: gradient domination of sa}), followed by global convergence of policy gradient over FIR policy classes (\Cref{thm app: global convergence}). 
\end{enumerate}

\section{The landscape of a single FIR-LQR cost}
\label{s app: landscape of single fir-lqr}
The convergence analysis of the sample average LQR objective builds on coercivity and gradient domination of a single LQR cost \cite{fujinami2025policygradientlqrdomain}. Since the lifted instance does not immediately have the required structure, we first introduce exact reformulations of the FIR-LQR cost that remove the degeneracy of the noise covariance and the state penalty matrices. Throughout this section, we drop the superscript $i$ when the argument is about a single cost. 

\subsection{Removing the degeneracy}
$x_t$ for all $t\geq1$ can be represented as a linear combination of the past noise $w_{t-l}, \forall l\leq t$, $\forall i \in\curly{1, \cdots, M}$, 
\begin{align}
    x_t = \sum_{l=0}^{t-1} \Phi_lw_{t-1-l}, \label{eq: closed-loop state response}
\end{align}
where we define the closed-loop state impulse response $\{\Phi_l^i\}_{l\geq0}$ by
\begin{align}
    \Phi_0 = I, \quad \Phi_l = (A+BK_0)\Phi_{l-1} + \sum_{h=1}^{H-1}BK_h\Phi_{l-1-h}, \quad l\geq 1, \quad \Phi_r=0(r<0). \label{eq app: closed-loop response}
\end{align}
% where $A_{K_0}^i = A^i+B^iK_0$ and with the convention $\Phi_r^i = 0$ for all $r < 0$. 
We also define the $H$-step noise controllability factor $C_K\coloneqq\brac{\bbB_w ~ \bbA_K\bbB_w ~ \cdots ~ \bbA_K^{H-1}\bbB_w}$. 
First, we show that we can remove the degeneracy of a single FIR-LQR cost by introducing exact reformulations with the non-degenerate H-step state cost matrix $\Sigma_H(K)$. 
\begin{lemma}[Structures of the lifted problem; proves \Cref{lem: reparameterized lifted lqr}]
    \label{lem app: strucure of lifted problem}
    Fix any $K$. Then
    \begin{enumerate}[label=(\alph*)]
        \item For every $t\ge0$, $\bbA_K^t\bbB_w = \brac{\Phi_t; \Phi_{t-1}; \cdots; \Phi_{t-H+1}}$. 
        \item $C_K$ is unit upper triangular matrix; in particular it is invertible and for every $K$
        \[
            \Sigma_H(K) = C_KC_K^\top \succ 0
        \]
        \item If $\rho(\bbA_K)<1$, the cost admits the exact reformulations
        \begin{align}
            J(K) = \trace\paren{(\bar\bbQ + K^\top RK)\Sigma_K} = \trace\paren{\bar{P}_K\Sigma_H(K)} \label{eq app: lifted cost}
        \end{align}
        where $\Sigma_K=\bbA_{K,H}\Sigma_K\bbA_{K,H}^\top+\Sigma_H(K)\succeq\Sigma_H(K)$. Moreover $P_K = \sum_{t=0}^{H-1}\bbA_K^{t\top}\bar{P}_K\bbA_K^t$ and $J(K) = \trace(P_K\bbW)$. 
    \end{enumerate}
\end{lemma}
\begin{proof}
    (a) Prove by induction. For $t=0$, $\bbB_w = \brac{\Phi_0; \Phi_{-1}; \cdots; \Phi_{-H+1}}$ from \eqref{eq app: closed-loop response}. For $t\ge1$, the first block row of $\bbA_K$ is $\begin{bmatrix}
            A + BK_0 & BK_1 & \cdots & BK_{H-1}
    \end{bmatrix}$
    while the remaining rows are the block shift matrix. Hence the first block row equals $(A+BK_0)\Phi_t + \sum_{h=1}^{H-1}BK_h\Phi_{t-h} = \Phi_{t+1}$, and the lower blocks are the shifted copies of $\Phi_{t}, \cdots, \Phi_{t-H+2}$. 

    (b) By (a), the $(j,t)$ block of $C_K$, $j,t\in\curly{0, \cdots, H-1}$ is $\Phi_{t-j}$: it vanishes for $t<j$ and equals $I$ for $t=j$. Thus $C_K$ is block unit upper triangular
    \[
        C_K = \begin{bmatrix}
            I & \Phi_1 & \Phi_2 & \cdots & \Phi_{H-1} \\
            0 & I & \Phi_1 & \cdots & \Phi_{H-2} \\
            0 & 0 & I & \cdots & \Phi_{H-3} \\
            \vdots & \vdots & \vdots & \ddots & \vdots\\
            0 & 0 & 0 & \cdots & I
        \end{bmatrix},
    \]
    hence invertible. It follows that $\Sigma_H(K) = \sum_{t=0}^{H-1}\bbA_K^t\bbB_w\bbB_w^\top\bbA_K^{t\top} = C_KC_K^\top\succ 0$. 

    (c) By (a), the $j$th diagonal block of $\Sigma_K = \sum_{t\ge0}\bbA_K^t\bbB_w\bbB_w^\top\bbA_K^{t\top}$, $j\in\curly{0, \cdots, H-1}$ is $\Sigma_x\coloneqq\sum_{t\ge0}\Phi_{t-j}\Phi_{t-j}^\top$ in the steady state. Hence $\trace(\bbQ\Sigma_K) = \trace(Q\Sigma_x) = \sum_{j=0}^{H-1}\trace\paren{\frac{1}{H}Q\Sigma_x} = \trace\paren{\bar\bbQ\Sigma_K}$, which gives the first equality in \eqref{eq app: lifted cost}. Splitting $\Sigma_K$ into groups of $H$ consecutive powers yields 
    \[
        \Sigma_K = \sum_{t\ge0}\bbA_K^t\bbW\bbA_K^{t\top} = \sum_{l\ge0}\bbA_{K,H}^l\Sigma_H(K)\bbA_{K,H}^{l\top} = \bbA_{K,H}\Sigma_K\bbA_{K,H}^\top + \Sigma_H(K)\succeq\Sigma_H(K).
    \]
    Substituting this series into the lifted cost and swapping the convergent sums from the cyclic property of trace gives $J(K) = \trace\paren{\bar{P}_K\Sigma_H(K)}$. The same $H$-step splitting of $P_K$ gives $P_K = \sum_{t\ge0}\bbA_K^{t\top}(\bar\bbQ + K^\top RK)\bbA_K^{t} = \sum_{t=0}^{H-1}\bbA_K^{t\top}\bar{P}_K\bbA_K^{t}$, and $\trace\paren{(\bar\bbQ+K^\top RK)\Sigma_K} = \trace\paren{P_K\bbW}$ follows from the standard Lyapunov duality. 
\end{proof}

\subsection{Coercivity, smoothness, and descent}
Next, we characterize the landscape of a single FIR-LQR cost, such as coercivity and smoothness, extending the corresponding LQR results \cite{hu2023toward}. 
\begin{lemma}[The cost bounds on the gain and the state excitation]
    \label{lem app: cost bounds on the FIR gain}
    Let $\rho(\bbA_K)<1$, and assume $J(K)\le\bar{J}$ with $\bar{J}\ge1$. Then
    \begin{enumerate}[label=(\alph*)]
        \item $\norm{\Phi_l}\le\sqrt{\bar{J}}$ for all $l$. 
        \item $\norm{K}_F \le\kappa(\bar{J})\coloneqq (1+\sqrt{\bar{J}})^H-1$. 
        \item $\lmin(\Sigma_K)\ge\lmin(\Sigma_H(K))\ge c(\bar{J})\coloneqq \paren{\sum_{h=0}^{H-1}(H\bar{J})^{h/2}}^{-2}.$
    \end{enumerate}
    Note that $\kappa(\cdot)$ is increasing and $c(\cdot)$ is decreasing. 
\end{lemma}
\begin{proof}
    (a) From \Cref{lem app: strucure of lifted problem} (c) and $Q\succeq I$, $\bar{J}\ge\trace(Q\Sigma_x)\ge \trace\Sigma_x = \sum_{l}\norm{\Phi_l}_F^2 \ge \norm{\Phi_l}^2$. 
    
    (b) By \Cref{lem app: strucure of lifted problem} (a), $K\bbA_K^l\bbB_w = \sum_{h=0}^{\min(l,H-1)}K_h\Phi_{l-h} \eqqcolon G_l$. From $R=I$ and $\Sigma_K\succeq\Sigma_H$ by \Cref{lem app: strucure of lifted problem} (c), $\bar{J}\ge\trace\paren{K\Sigma_KK^\top}\ge\norm{KC_K}_F^2 = \sum_{l=0}^{H-1}\norm{G_l}^2_F$, so $\norm{G_l}_F\le\sqrt{\bar{J}}$. Since $G_l = K_l + \sum_{h<l}K_h\Phi_{l-h}$, by the reverse triangle inequality $\norm{K_l}_F\le\sqrt{\bar{J}}(1+\sum_{h<l}\norm{K_h}_F)$. Now we  claim $\norm{K_l}_F \leq \sqrt{\bar{J}}\paren{1+\sqrt{\bar{J}}}^l$ by induction. It follows that
    \begin{align*}
        \norm{K_l}_F &\leq \sqrt{\bar{J}}(1+\sum_{h<l}\norm{K_h}_F) \leq \sqrt{\bar{J}} + \sqrt{\bar{J}}\sum_{k=0}^{l-1}\sqrt{\bar{J}}\paren{1+\sqrt{\bar{J}}}^k \\
        % &= \sqrt{J^i(K)} + J^i(K)\sum_{k=0}^{l-1}\paren{1+\sqrt{J^i(K)}}^k \\
        &= \sqrt{\bar{J}} + \bar{J}\frac{(1+\sqrt{\bar{J}})^l-1}{\sqrt{\bar{J}}}  = \sqrt{\bar{J}}\paren{1+\sqrt{\bar{J}}}^l,
    \end{align*}
    where the second inequality is from the induction assumption. Thus $\norm{K}_F = \left(\sum_{l<H}\norm{K_l}^2_F\right)^{1/2} 
    \leq \sum_{l<H}\norm{K_l}_F \leq \sqrt{\bar{J}}\sum_{l<H}\paren{1+\sqrt{\bar{J}}}^l = \paren{1+\sqrt{\bar{J}}}^{H}-1$ where the last equality follows from the finite geometric series. 

    (c) Write $C_K = I+U$. By \Cref{lem app: strucure of lifted problem} (b), U is strictly block upper triangular and block Toeplitz matrix with blocks $\Phi_1, \cdots, \Phi_{H-1}$. By the triangle and Cauchy-Schwarz inequality, $\norm{U}\le\sum_{l=1}^{H-1}\norm{\Phi_l}\le\sqrt{(H-1)\sum_{l\ge1}\norm{\Phi_l}^2}\le\sqrt{H\bar{J}}$. Since $U^H=0$, $C_K^{-1} = \sum_{h=0}^{H-1}(-U)^h$, so $\norm{C_K^{-1}}\le\sum_{h<H}(H\bar{J})^{h/2}$. By (b)-(c) of \Cref{lem app: strucure of lifted problem}, $\lmin(\Sigma_K)\ge\lmin(\Sigma_H)=\sigma_{min}(C_K)^2 \ge c(\bar{J}) \coloneqq (\sum_{h<H}(H\bar{J})^{h/2})^{-2}$. 
\end{proof}

\begin{lemma}[Coercivity, compact sublevel sets and smoothness% of sample average FIR-LQR Objective, 
, extended version of \Cref{lem: coercivity of sa FIR-LQR}]
    \label{lem app: coercivity of sa FIR-LQR}
    \addcontentsline{toc}{subsubsection}{Lemma \protect\ref{lem: coercivity of sa FIR-LQR appendix}: Coercivity and Smoothness}
    For $i\in\brac{M}$, $J^i(K)$ is coercive on $\calK^i$, i.e. for any sequence $\{K^l\}_{l=1}^\infty\subset\calK^i$ we have $J^i(K^l)\rightarrow\infty$ if either $\norm{K^l}_F\rightarrow\infty$, or $K^l\rightarrow K$ with $\rho(\bbA_K^i)=1$. Consequently, $\sa(K)$ is coercive on $\calK$, every sublevel set $\calK_\zeta$ is compact, and $\sa(K)$ is twice continuously differentiable and $L$-smooth on $\calK_{\zeta}$. Moreover, $\sa(K)$ attains its minimum on $\calK$ at an interior point $K^\star$ with $\nabla\sa(K^\star) = 0$. 
\end{lemma}
% \tas{this is not a complete proof as it requires bounding the least singular value of Phi (which in turn depends on the norm of Phi since the diagonal is identity), which is not discussed. This is one of the technical challenges in this paper since we have to bound the norms of $\Phi$ and $K$ simultaneously. We should either remove the proof environment and discuss informally or do a proper full proof}
\begin{proof}
From \Cref{lem app: cost bounds on the FIR gain} (b), the cost goes to infinity as $\norm{K^l}_F\rightarrow\infty$. For the stability boundary, define $\phi(J) \coloneqq J/c(J)$ which is an increasing function of $J$ and continuous on $[1, \infty)$. From \Cref{lem app: strucure of lifted problem} (c), $\Sigma_{K^l}^i \succeq \sum_{m=0}^N (\bbA_{K^l,H}^i)^m\Sigma_H(\bbA_{K^l,H}^i)^{m\top}$ for any $N\in\mathbb{N}$. 
Thus from \Cref{lem app: cost bounds on the FIR gain} (c) and $\bar\bbQ + K^\top RK\succeq\frac{1}{H}I$, 
\[
    J^i(K^l)\ge\frac{1}{H}\trace\paren{\Sigma_{K^l}^i}\ge\frac{c(J^i(K^l))}{H}\sum_{m=0}^N\norm{(\bbA_{K^l,H}^{i})^m}^{2}_F \iff \phi(J^i(K^l)) \ge \frac{1}{H}\sum_{m=0}^N\norm{(\bbA_{K^l,H}^{i})^m}^{2}_F
\]
Suppose $K^l\rightarrow K$ with $\rho(\bbA_K^i)=1$, i.e. $K^l$ approaches the stability boundary. Each term in the sum converges to $\norm{(\bbA_{K,H}^{i})^m}_F \ge \rho(\bbA_{K,H}^{i})^m = \rho(\bbA_K^i)^{Hm} = 1$. Hence, $\frac{1}{H}\sum_{m=0}^N\|(\bbA_{K^l,H}^{i})^m\|^{2}\rightarrow (N+1)/H$. Since $N$ is arbitrary, $\phi(J^i(K^l)) \rightarrow \infty$. From monotonicity of $\phi$, $J^i(K^l)\rightarrow\infty$. Therefore, $J^i(K)$ is coercive on $\calK^i$.
Then, since $M\sa(K) \geq J^i(K)$, $\sa(K)$ is also coercive on $\calK=\cap_i\calK^i$ since leaving $\calK$ means leaving some $\calK^i$. Moreover, from \cite[Lemma~1]{hu2023toward}, $J^i(K)$ is twice continuously differentiable $C^2$ and thus $\sa(K)$ is also $C^2$. Therefore, from \cite[Thm.~1]{hu2023toward}, $\calK_{\zeta}$ is compact, and $\norm{\nabla_K^2 \sa(K)}$ is bounded on $\calK_{\zeta}$. 
Let this uniform upper bound as $L$, and hence $\sa(K)$ is $L$-smooth and satisfy the following inequality by the mean value theorem $\sa(K') - \sa(K) \leq \trace\paren{(K'-K)^\top \nabla_K \sa(K)} + \frac{L}{2}\norm{K'-K}_F^2$
for any $K, K'\in \calK_{\zeta}$. $L$ is determined as
\[
    L(\zeta) \coloneqq \max_{K\in\calK_\zeta}\norm{\nabla^2 \sa(K)}\le \max_{i\in\brac{M}}\max_{K\in\calK_\zeta}\norm{\nabla^2 J^i(K)} < \infty
\]
which depends on $\zeta$ and the problem parameters only. 
Also, since $\calK_{\zeta}$ is compact and $\calK$ is open, the minimum $K^\star\in\text{int}(\calK)$ with $\nabla\sa(K^\star) = 0$. 

% From \Cref{lem: K bound}, $J^i(K)$ is lower bounded by $\norm{K}^2$. Also, for any $K\in\calK^i$, $J^i(K)$ goes to infinity when $K$ approaches the boundary of stability region \tas{proof?}\Tesshu{detectability of $\bar{\bbQ}$ can justify}. Therefore, it holds that $J^i(K)\rightarrow \infty$ if $\norm{K}\rightarrow\infty$ or $K$ converges to an element of $\partial \calK^i$, and thus $J^i(K)$ is coercive. 
% Then, since $M\sa(K) \geq J^i(K)$, $\sa(K)$ is also coercive on the simultaneously stable region $\calK$ defined in Assumption \ref{asmp-app: simultaneous stabilization}. 
% Moreover, from \cite[Lemma~1]{hu2023toward}, $J^i(K)$ is twice continuously differentiable $C^2$ and thus $\sa(K)$ is also $C^2$. Therefore, from \cite[Thm.~1]{hu2023toward}, $\calK_{\zeta}$ is compact and thus $\norm{\nabla_K^2 \sa(K)}$ is bounded on $\calK_{\zeta}$. Let this uniform upper bound as $L$, and hence $\sa(K)$ is $L$-smooth and satisfy the following inequality by the mean value theorem
% \[
%     \sa(K') - \sa(K) \leq \trace\paren{(K'-K)^\top \nabla_K \sa(K)} + \frac{L}{2}\norm{K'-K}_F^2
% \]
% for any $K, K'\in \calK_{\zeta}$. 
\end{proof}

% \tas{we need to argue about the existence of an $L$ that is independent of the current iteration $K$}

Following \Cref{lem app: coercivity of sa FIR-LQR}, gradient descent converges to a fixed point. This is a direct consequence of \cite[Thm.~1]{hu2023toward}, which applies to $\sa(K)$ by coercivity and $L$-smoothness. 
% \tas{Is the statement of Theorem 1 in Hu et al for general J(K) with stated properties? Or is it for the specific LQR J(K)?}\Tesshu{mention that the proof follows the same structure}
\begin{lemma}[Convergence to a Fixed Point]
    \label{lem app: convergence to a fixed point}
    \addcontentsline{toc}{subsubsection}{Lemma \protect\ref{lem: convergence to a fixed point appendix}: Convergence to a fixed point}
    Let $\zeta_0 = J_{SA}(K^{(0)})$ and $L$ be the $L$-smoothness constant of $J_{SA}(K)$ on $\mathcal{K}_{\zeta_0}$.
    Consider the policy gradient method \eqref{eq: sa pg}.  Then, for any $0 < \alpha < \frac{2}{L}$, we have $K^{(n)}\in\mathcal{K}_{\zeta_0}$ and $J_{SA}(K^{(n+1)}) \leq J_{SA}(K^{(n)})$ for all n.
    Furthermore, we have the convergence of the gradient $\nabla_K J_{SA}(K) \rightarrow 0$ with the rate
    \begin{align}
        \min_{0\leq l\leq k}\norm{\nabla_K J_{SA}(K^{(l)})}^2_F \leq \frac{\zeta_0}{C(k+1)}, \label{eq app: grad convergence rate}
    \end{align}
    with $C = \alpha - \frac{L\alpha^2}{2} > 0$.
\end{lemma}
% \tas{notation should be consistent with (5) where you use $K^{(i)}$. Also $K_0$ clashes with the notation for $K=[K_0\dots]$, do you mean $K^{(0)}$ here?}
% Note that from \cite[Theorem 7]{fazel2018global} \tas{I couldn't find anything in the statement of Theorem 7 of fazel about the smoothness constant, did you extract it from the proof? Also here you have average of costs so it is not immediate why you would get that} the smoothness constant $L$ can be explicitly written as a function of the problem parameters over the sublevel set of $\sa(K^{(0)})$, i.e. 
% \begin{align*}
%     L \leq f_L(&\sup_{i\in\{1, \cdots, M\}}J^i(K^{(0)}),\sup_{i\in\{1, \cdots, M\}}\norm{\nabla_KJ^i(K_0)}, H, \sup_{i\in\{1, \cdots, M\}}\norm{A^i},\sup_{i\in\{1, \cdots, M\}}\norm{B^i})
% \end{align*}
% \tas{would it be easier to say $L_i$ is the smoothness of system $i$ bounded by... Can you then say that $L\le \max L_i$?}
% \tas{In any case it should be clarified what $f_L$ means}

\subsection{Gradient domination, quadratic growth}
Gradient domination is the key to global convergence in the non-convex optimization. We show that the FIR-LQR cost is gradient dominated, extending the corresponding LQR result \cite{fazel2018global}. 
\begin{lemma}[Cost difference]
    Let $K, K'$ be stabilizing gains for the lifted system, and let $\Delta \coloneq K' - K$. Then 
    \begin{align}
        \label{eq app: cost difference}
        J(K') - J(K) = \trace\paren{\Sigma_{K'}\brac{2\Delta^\top E_K + \Delta^\top (R+\bbB^\top P_K\bbB)\Delta}}
    \end{align}
\end{lemma}
\begin{proof}
    From $J(K) = \trace(P_K\bbW)$, $\Sigma_{K'} = \bbA_{K'}\Sigma_{K'}\bbA_{K'}^\top + \bbW$ and the cyclic property of trace, $J(K) = \trace(P_K[\Sigma_{K'} - \bbA_{K'}\Sigma_{K'}\bbA_{K'}^\top]) = \trace(\Sigma_{K'}[P_K - \bbA_{K'}^\top P_K\bbA_{K'}])$. Subtracting this from $J(K')=\trace(\Sigma_{K'}(\bar\bbQ + K'^\top RK'))$, 
    \[
        J(K') - J(K) = \trace(\Sigma_{K'}M), \quad M\coloneqq \bar\bbQ + K'^\top RK' - P_K + \bbA_{K'}^\top P_K\bbA_{K'}
    \]
    Expanding $\bbA_K' = \bbA_K + B\Delta$, $K' = K + \Delta$, and canceling $-P_K + \bbA_K^\top P_K\bbA_K = -(\bar\bbQ + K^\top RK)$,
    \[
        M = \Delta^\top(RK + \bbB^\top P_K\bbA_K) + (RK + \bbB^\top P_K\bbA_K)^\top \Delta + \Delta^\top (R + \bbB^\top P_K\bbB) \Delta
    \]
    Observing $RK + \bbB^\top P_K\bbA_K = E_K$, and $\trace(\Sigma_{K'}\Delta^\top E_K) =  \trace(\Sigma_{K'}E_K^\top\Delta)$ from symmetry of $\Sigma_{K'}$ yield \eqref{eq app: cost difference}. 
\end{proof}

\begin{lemma}[Gradient domination and quadratic growth of a single cost]
    \label{lem: gradient dominance of lqr}
    \addcontentsline{toc}{subsubsection}{Lemma \protect\ref{lem: gradient dominance of lqr}: Gradient dominance of FIR-LQR}
    Let $K$ be a stabilizing gain, and $K^\star$ be the minimizer of $J(K)$. Write $c_{K,H} \coloneqq c(J(K))$. Then
    \begin{enumerate}[label=(\alph*)]
        \item $J(K) - J(K^\star) \le \frac{HJ(K^\star)}{4c_{K,H}^2}\norm{\nabla J(K)}_F^2$
        \item $J(K) - J(K^\star) \ge c_{K,H}\norm{K-K^\star}_F^2$.
    \end{enumerate}
    % Let $K$ and $K^\star$ be an arbitrary stabilizing and optimal FIR controller, and $J(K)$ denote a single LQR cost with a controller $K$. 
    % % \tas{define what J is here} Then it holds that
    % \[
    %     J(K) - J(K^\star) \leq \frac{HJ(K^\star)}{c_{K,H}^2}\norm{\nabla J(K)}_F^2
    % \]
\end{lemma}
\begin{proof}
    (a) Completing the square in \eqref{eq app: cost difference} with $K' = K^\star$ (cf. \cite[Lemma~11]{fazel2018global}) gives
    \begin{align*}
        J(K) - J(K^\star) 
        &\leq \trace\paren{\Sigma_{K^\star}E_K^{\top}(R+\bbB^{\top} P_K\bbB)^{-1}E_K} 
        \leq \frac{\norm{\Sigma_{K^\star}}}{\sigma_{min}(R)}
        \trace\paren{E_K^\top E_K} 
        \le \norm{\Sigma_{K^\star}}\norm{E_K}_F^2
    \end{align*}
    using $R=I$. By \Cref{lem app: cost bound} (a) below, $\|\Sigma_{K^\star}\|\le HJ(K^\star)$, and $\nabla J(K) = 2E_K\Sigma_K$ with \Cref{lem app: cost bounds on the FIR gain} (c) gives $\|E_K^i\|_F\le\|\nabla J^i(K)\|_F/2c_{K,H}$. 

    (b) Let $\Delta_\star \coloneq K - K^\star$. 
    Since $E_{K^\star}=0$, \eqref{eq app: cost difference} with base point $K^\star$ reads 
    \[
        J(K) - J(K^\star) = \trace(\Sigma_K\Delta_\star^\top(R+\bbB^{\top}P_K^\star\bbB)\Delta_\star) = \VEC\Delta_\star^\top\paren{\Sigma_K\kron(R+\bbB^{\top}P_K^\star\bbB}\VEC\Delta_\star \ge c_{K,H}\|\Delta_\star\|_F^2
    \]
    from $R=I$. 
\end{proof}

\subsection{Optimal FIR policy}
\label{s app: optimal fir policy}
Memory carries no benefit for a single system; this yields the reference controllers used to compare systems in the next section. Let $\Psta$ denote the solution to the static discrete algebraic Riccati equation (or \dare) of $(A,B,Q,R)$ and $\Ksta = -(R+B^\top \Psta B)^{-1}B^\top \Psta A$ the static LQR gain. 
\begin{lemma}[Lifted $\dare$]
     \label{lem app: lifted dare}
    \addcontentsline{toc}{subsubsection}{Lemma \protect\ref{lem: lifted dare}: Lifted DARE}
    Let $P^\star$ be the solution to $\dare$ of the lifted instance $(\bbA, \bbB, \bar\bbQ, R)$. Then it holds that
    \begin{align}
        P^\star = \diag\paren{\Psta, ~ p_1Q, \cdots, ~ p_{H-1}Q}, \quad p_j = \frac{H-j}{H}
        \label{eq: P-bar lifted dare}
    \end{align}
    with optimal gain $\Ksingle^\star = \begin{bmatrix}
            \Ksta^\star & 0 & \cdots & 0
    \end{bmatrix}$, and the FIR optimal cost equals the optimal static cost: $J(\Ksingle^\star) = \trace\paren{P^\star\bbW} = \trace(\Psta)$. 
\end{lemma}
\begin{proof}
    Assume ${P^\star}$ has the diagonal form as in \eqref{eq: P-bar lifted dare}. $\bbB^\top {P^\star}\bbB = B^\top \Psta B$ and $\bbB^\top {P}\bbA = \begin{bmatrix}
            B^\top \Psta A & 0 & \cdots & 0
    \end{bmatrix}$, so the $\dare$ gain is $\begin{bmatrix}
            \Ksta^\star & 0 & \cdots & 0
    \end{bmatrix}$. Since $\bbA^\top {P}\bbA = \diag\paren{A^\top \Psta A + p_1Q, p_2Q, \cdots, p_{H-1}Q, 0}$, the Riccati recursion $P^\star = {\bbQ} + \bbA^\top {P^\star}\bbA - \bbA^\top {P^\star}\bbB(R + \bbB^\top {P^\star}\bbB)^{-1}\bbB^\top {P^\star}\bbA$ decouples: block $(0,0)$ reads \[\frac{1}{H}Q + A^\top \Psta A + p_1Q - A^\top \Psta B(R+B^\top \Psta B)^{-1}B^\top \Psta A = \Psta\] which is the static $\dare$ since $p_1Q + \frac{1}{H}Q = Q$; block $(j,j), j\ge1$ reads $p_{j+1}Q + \frac{1}{H}Q = p_jQ$. Also, $\bbW = \diag(I, 0, \cdots, 0)$ gives $\trace\paren{P^\star\bbW} = \trace(\Psta)$. 
    
    % \begin{align*}
    %     &\bbA^\top {P}\bbA = \diag\paren{A^\top \Psta A + p_1Q, p_2Q, \cdots, p_{H-1}Q, 0} \\
    %     &\bbB^\top {P}\bbA = \paren{\bbA^\top {P}\bbB}^\top = \begin{bmatrix}
    %         B^\top \Psta A & 0 & \cdots & 0
    %     \end{bmatrix} \\
    %     &\bbB^\top {P}\bbB = B^\top \Psta B \\
    %     &\bbA^\top {P}\bbB(R + \bbB^\top {P}\bbB)^{-1}\bbB^\top {P}\bbA = \diag\paren{A^\top \Psta B(R+B^\top \Psta B)^{-1}B^\top \Psta A, 0, \cdots, 0}
    % \end{align*}
    % By summarizing, it follows from \eqref{eq: dare to lifted lqr} that
    % \begin{align*}
    %     {P} &= \diag\paren{\frac{1}{H}Q, \cdots, \frac{1}{H}Q} + \diag\paren{A^\top \Psta A + p_1Q, p_2Q, \cdots, p_{H-1}Q, 0} \\
    %     & \hspace{5mm} + \diag\paren{A^\top \Psta B(R+B^\top \Psta B)^{-1}B^\top \Psta A, 0, \cdots, 0} \\
    %     &= \diag\paren{Q + A^\top \Psta A + A^\top \Psta B(R+B^\top \Psta B)^{-1}B^\top \Psta A, p_1Q, \cdots, p_{H-1}Q} \\
    %     &= \diag\paren{\Psta , p_1Q, \cdots, p_{H-1}Q}
    % \end{align*}
    % where the last equality follows from $\dare$ of LQR. Thus we conclude that ${P}$ has the diagonal form in \eqref{eq: P-bar lifted dare}. 

    % Furthermore, $K^\star$ satisfies
    % \begin{align*}
    %     K^\star = -(R+\bbB^\top {P}\bbB)^{-1}\bbB^\top {P}\bbA = -(R+B^\top \Psta B)^{-1}\begin{bmatrix}
    %         B^\top \Psta A & 0 & \cdots & 0
    %     \end{bmatrix} = \begin{bmatrix}
    %         \Ksta^\star & 0 & \cdots & 0
    %     \end{bmatrix}
    % \end{align*}
\end{proof}

We further show that the diagonal structure of the $\dare$ solution $P$ remains as long as the controller $K$ only has a first element. 
\begin{lemma}[Policy evaluation with a structured gain]
    \label{lem app: policy evaluation with the optimal gain}
    \addcontentsline{toc}{subsubsection}{Lemma \protect\ref{lem: policy evaluation with the optimal gain}: Policy Evaluation with the Optimal Gain}
    Let $K = \begin{bmatrix}
            \Ksta & 0 & \cdots & 0
        \end{bmatrix}$.
    Then with $\Psta(\Ksta) \coloneqq \dlyap(A+B\Ksta, Q+\Ksta^\top R\Ksta)$, it holds that
    \begin{align}
        P_K = \diag\paren{\Psta(\Ksta), ~ p_1Q, \cdots, ~ p_{H-1}Q}, 
        \quad p_j = \frac{H-j}{H}
        \label{eq: lifted P_K diagonal}
    \end{align}
    Consequently, $J(K) = \trace\paren{P_K\bbW} = \trace(\Psta(\Ksta))$. 
\end{lemma}
\begin{proof}
    Verify the Lyapunov equation $P_K = \bbA_K^\top P_K\bbA_K + \bar\bbQ + K^\top RK$ with the block diagonal form \eqref{eq: lifted P_K diagonal}.
    First it follows that $\bar\bbQ + K^\top RK = \diag\paren{\Ksta^\top R\Ksta + \frac{Q}{H}, ~ \frac{Q}{H}, \cdots, ~ \frac{Q}{H}}$.
    Second from the structure of $\bbA$ and $\bbB$, $\bbA_K^\top P_K\bbA_K = \diag\paren{A_{\Ksta}^\top \Psta(\Ksta)A_{\Ksta} + p_1Q, ~ p_2Q, \cdots, ~ p_{H-1}Q, ~ 0}$. 
    Thus the first diagonal element of $P_K$ becomes
    \[
        A_{\Ksta}^\top \Psta(\Ksta)A_{\Ksta} + p_1Q + \Ksta^\top R\Ksta + \frac{Q}{H} = \Psta(\Ksta),
    \]
    and the rest of elements follow $p_jQ + \frac{Q}{H} = p_{j-1}Q$
    from the definition. 
    % Thus we conclude that $P_K$ has the diagonal form in \eqref{eq: lifted P_K diagonal}. 
\end{proof}

\begin{lemma}[$\dare$ inequalities]
    \label{lem app: dare inequalities}
    \addcontentsline{toc}{subsubsection}{Lemma \protect\ref{lem: useful inequalities for dare}: Useful inequalities for $\dare$}
    For any $\theta = (A,B)\in\calS$ with static $\dare$ solution $\Psta$ and gain $\Ksta$, and $K = \brac{\Ksta ~ 0 \cdots 0}$, the following inequalities hold: 
    \begin{enumerate}[label=(\alph*)]
        \item $1\le\norm{P^\star} = \norm{\Psta}\le\trace(\Psta) = J^\theta(K(\theta))\le\bar{J}_\star$
        \item $\norm{\Ksta} \leq \norm{\Psta}^{1/2}$ and $\norm{A+B\Ksta}\le\norm{\Psta}^{1/2}$
        \item $\norm{\bbA_K}^2\le\norm{\Psta} + 1$
    \end{enumerate}
\end{lemma}
\begin{proof}
    (a) Since $P$ takes the diagonal structure from \eqref{eq: P-bar lifted dare},$\norm{P} = \max\curly{\norm{\Psta}, p_1\norm{Q}, \cdots, p_{H-1}\norm{Q}} = \norm{\Psta}$
    since $p_i\le1$ for $i=1, \cdots, H-1$ and $\norm{Q}\le\norm{\Psta}$. $Q\succeq I$ yields the first inequality. 
    (b) follows from Lemma 6 in \cite{fujinami2025domainrandomizationsampleefficient} with $R=I$. 
    (c) $\bbA_K^\top \bbA_K$ is the block diagonal with blocks $(A+B\Ksta)^\top(A+B\Ksta)+I$ in $(0,0)$, $I$ in $(j,j)$, $j = 1, \cdots, H-2$, and $0$. Thus $\norm{\bbA_K}^2 = \norm{A+B\Ksta}^2+1\le\norm{\Psta}+1$. 
    % follows from 
    % \begin{align*}
    %     \norm{\bbA_{K}}^2 &= \lambda_{max}\paren{(\bbA_{K})^\top\bbA_{K}} = \lambda_{max}\brac{\diag\paren{(A_{\Ksta^i}^i)^\top A_{\Ksta^i}^i + I, I, \cdots, I, 0}} \\
    %     &= \max\curly{\norm{A_{\Ksta^i}^i}^2+1, 1} = \norm{A_{\Ksta^i}^i}^2+1 \le \norm{\Psta^i}+1
    % \end{align*}
\end{proof}

\section{Perturbation argument on the FIR-LQR cost}
\label{s app: perturbation argument on the fir-lqr cost}
Given \Cref{s app: landscape of single fir-lqr}, the goal is to derive gradient domination of the sample average FIR-LQR cost, which holds only when the heterogeneity gap between the samples is small. To formalize this, we develop perturbation bounds on the closed-loop system-theoretic quantities. All named constants are summarized in \Cref{tab:constants}.   
\begin{table}[h!]
\centering
{\footnotesize
\renewcommand{\arraystretch}{1.45}
\begin{tabular}{@{}l@{\;\;}l@{\quad}l@{}}
\toprule
 & Definition & Role (where established) \\
\midrule
$\bar{J}_\star$ & $\sup_{\theta\in\calS}J^{\theta}(K(\theta))$,  & optimal cost bound over $\calS$ (\eqref{eq app: sup optimal gain}, \Cref{lem app: dare inequalities})\\
$\tilde{J}$ & $32\,\Jbs$ & uniform cost bound on $\calK_{\zeta_0}$ (\Cref{cor:uniform})\\
$\bar{c}$ & $\paren{\sum_{h=0}^{H-1}(H\tJ)^{h/2}}^{-2}$ & uniform excitation: $\lmin(\Sigma^i_K)\ge\cb$ (\Cref{lem app: cost bounds on the FIR gain}, \Cref{cor:uniform})\\
% $\kb$ & $(1+\sqrt{\tJ}\,)^{H}-1$ & uniform gain bound: $\norm{K}_F\le\kb$ (\Cref{cor:uniform})\\
$\nb$ & $2\,\sqrt{\tJ/\cb}$ & uniform bound: $1+\norm{K}_F\le\nb$ (\Cref{lem app: cost bound}(d))\\
\midrule
$C_A$ & $H^{2}\tJ/\cb$ & closed-loop perturbation constants (\Cref{lem app: perturbation bound} (b))\\
$C_\Sigma$ & $12\,H^{9/2}\tJ^{4}/\cb^{3}$ & state covariance perturbation constants (\Cref{lem app: perturbation bound} (c))\\
$C_P$ & $8\,H^{7/2}\tJ^{4}/\cb^{4}$ & value matrix perturbation constants (\Cref{lem app: perturbation bound} (d))\\
$C_E$ & $13\,\tau_B\,H^{11/2}\tJ^{5}/\cb^{5}$ & gradient factor $E_K$, model direction (\Cref{lem app: E_K perturbation})\\
$C_{G,1}$ & $2H\tJ\,C_E$ & gradient, model direction: absolute part (\Cref{lem app: gradient perturbation})\\
$C_{G,2}$ & $2C_\Sigma/\cb$ & gradient, model direction: relative part (\Cref{lem app: gradient perturbation}) \\
$\Lambda_1$ & $2\,C_{G,1}$ & $\Lambda(K)=2\grad(K)+\Lambda_1\het$ (\Cref{lem app: gradient bound})\\
$C_\gamma$ & $4\,\tB\,H^{9/2}\tJ^{7/2}/\cb^{7/2}$ & gain perturbation constants (\Cref{lem app: gain perturbation})\\
\midrule
$\barhet^1$ & $\paren{54\,H\,\Jbs^{\,5}}^{-1}$ & cost bound with heterogeneity (Lem.~\ref{lem app: cost bound with heterogeneity})\\
$\barhet^2$ & $\cb^{\,3}\big/\big(12\,H^{7/2}\tJ^{3}\big)$ & scenario boundedness (Cor.~\ref{cor:uniform})\\
$\barhet^3$ & $(2\,C_{G,2})^{-1}=\cb^{\,4}\big/\big(48\,H^{9/2}\tJ^{4}\big)$ & gradient bound (\Cref{lem app: gradient bound})\\
$\barhet^{4}$ & $\cb^{\,2}\big/\big(2\,\Lambda_1H\tJ\,C_\gamma\big)$ & quadratic growth (\Cref{lem app: quadratic growth of sa})\\
% $\bar\eps_{\mathrm{grad}}$ & $\cb/(8\gamma)=1/f_{\mathrm{grad}}$, see \eqref{eq:fgrad} & gradient-domination region (Lem.~\ref{lem:gd})\\
% $\bar\eps_{\mathrm{het}}$ & $\min\big\{\bar\eps^{\,1}_{\mathrm{het}},\dots,\bar\eps^{\,4}_{\mathrm{het}},\,\cb/(4\gamma\Lambda_1)\big\}=1/f_{\mathrm{het}}$, see \eqref{eq:fhet} & Assumption~\ref{ass:het} (Lem.~\ref{lem:gd})\\
\bottomrule
\end{tabular}}
\caption{All named constants of this appendix. Each is an explicit function of $(H,\tau_B,\Jbs)$ only, independent of $M$, of the gain $K$, and of the iterates of \eqref{eq app: pg}. Numerical constants are not optimized; only the functional dependence matters for Theorem~\ref{thm app: global convergence}.}
\label{tab:constants}
\end{table}

\subsection{Cost and Lyapunov bounds}
We first introduce upper bounds on the closed-loop system-theoretic quantities in terms of the LQR cost and Lyapunov functions that are used throughout the rest of the proof. 

\begin{lemma}[Cost and Lyapunov bounds]
    \label{lem app: cost bound}
    \addcontentsline{toc}{subsubsection}{Lemma \protect\ref{lem: cost bound}: Cost bound}
    Let $K\in\calK^i$. Then it holds that
    \begin{enumerate}[label=(\alph*)]
        \item $\norm{{\Sigma}_K^i}_{\{\cdot,F\}} \leq HJ^i(K)$ 
        % \tas{isn't $\Sigma=\bar{\Sigma}$?}
        \item $\norm{K}^2 \leq \norm{\bar{P}_K^i}_{\{\cdot,F\}} \leq \frac{J^i(K)}{c_{K,H}}$ 
        \item $\norm{(\bbA_K^i)^t}_{\{\cdot,F\}} \leq \sqrt{\frac{HJ^i(K)}{c_{K,H}}}$ for $t\geq1$
        % \item $\norm{P_K} \leq C\norm{\bar{P}_K}$
        \item $1\leq\nu(K)\leq2\sqrt{\frac{J^i(K)}{c_{K,H}}}$
        % , \quad $1\leq\zeta(K)\leq\frac{2J^i(K)}{c_{K,H}}$\tas{do we need the new symbol for $\nu$ since it is simply $1+\|K\|^{1,2}$?} \tas{why not use (27) directly to bound them?}
        \item $\norm{P_K^i}\leq\frac{H^2J^i(K)^2}{c_{K,H}^2}$
    \end{enumerate}
    We abbreviate $c_{K,H} \coloneqq c(J(K))$. 
\end{lemma}
\begin{proof}
    (a) By \Cref{lem app: cost bounds on the FIR gain} (c) and $\bar\bbQ\succeq\frac{1}{H}I$, 
    $J^i(K) = \trace\paren{(\bar{\bbQ} + K^\top RK){\Sigma}_K^i} \geq \frac{1}{H}\trace(\Sigma_K^i)$. 
    Since $\Sigma_K^i\succ0$, $ \norm{\Sigma_K^i} \leq \norm{\Sigma_K^i}_F \leq \trace\paren{\Sigma_K^i} \leq HJ^i(K)$. 
    
    (b) The first inequality holds since $\bar{P}_K^i = \dlyap((\bbA_K^i)^H, \bar{\bbQ} + K^\top RK) \succeq K^\top K$. The second inequality follows from \Cref{lem app: strucure of lifted problem} (c) that $J^i(K) = \trace(\bar{P}^i_K\Sigma_H^i) \geq c_{K,H}\trace\paren{\bar{P}_K^i}$
    and $\bar{P}_K^i\succ0$ yields the result. 
    
    (c) Note that $(\bbA_K^i)^t\Sigma_K^i((\bbA_K^i)^t)^\top \preceq \Sigma_K^i - \bbW \preceq \Sigma_K^i$ for $t\geq1$. 
    Since $\Sigma_K^i \succ 0$, we get
    \[\norm{(\Sigma_K^i)^{-1/2}(\bbA_K^i)^t\Sigma_K^i((\bbA_K^i)^t)^\top(\Sigma_K^i)^{-1/2}} = \norm{(\Sigma_K^i)^{1/2}((\bbA_K^i)^t)^\top(\Sigma_K^i)^{-1/2}}^2 \leq 1
    \]
    Let $\Tilde{\bbA}_K^i = (\Sigma_K^i)^{1/2}((\bbA_K^i)^t)^\top(\Sigma_K^i)^{-1/2}$. It holds that
    \[
        \norm{(\bbA_K^i)^t}_{\{\cdot,F\}}  = \norm{(\Sigma_K^i)^{-1/2}\Tilde{\bbA}_K^i(\Sigma_K^i)^{1/2}}_{\{\cdot,F\}} \leq \norm{(\Sigma_K^i)^{-1/2}}\norm{\Tilde{\bbA}_K^i}\norm{(\Sigma_K^i)^{1/2}}_{\{\cdot,F\}} \leq \sqrt{\frac{\sigma_{max}(\Sigma_K^i)}{\sigma_{min}(\Sigma_K^i)}} \leq \sqrt{\frac{HJ^i(K)}{c_{K,H}}}
    \]
    where we used $\norm{ABC}_F\leq\norm{AB}\norm{C}_F\leq\norm{A}\norm{B}\norm{C}_F$ in the first inequality. 
    % \Tesshu{separate op and F}
    % \tas{should the middle part be
    % \[\norm{(\Sigma_K^i)^{-1/2}}_{op}\norm{\Tilde{\bbA}_K^i}_{\{op,F\}}\norm{(\Sigma_K^i)^{1/2}}_{op}
    % \]}\Tesshu{Can it be $\norm{ABC}_F\leq\norm{AB}\norm{C}_F\leq\norm{A}\norm{B}\norm{C}_F?$}
    % \tas{also what does $\|\cdot\|$  without subscripts denote for matrices? does it denote operator norm? if yes then we should make the notation consistent. Either go for op everywhere or remove it everywhere}
    % For the second part, it holds that
    % \[
    % \norm{\bbA_{K,H}^i} = \norm{\bbA_{K,H}(t)^\top \bbA_{K,H}(t)}^{1/2} \leq \sqrt{H}\norm{P(t)}^{1/2}
    % \]
    % since $H(\bbQ + K^\top RK) \succeq 1$.
    
    % For (d)
    % \begin{align*}
    %     P_K &= \sum_{t=0}^\infty (\bbA_K^t)^\top (\bar{\bbQ} + K^\top RK)\bbA_K^t 
    %     = \sum_{j=0}^{H-1}(\bbA_K^j)^\top\bar{P}_K\bbA_K^j \preceq \paren{\sum_{j=0}^{H-1} \norm{\bbA_K}^{2j}}\bar{P}_K 
    % \end{align*}
    (d) It follows from (b) that $1 \leq \nu(K) = 1 + \norm{K} \leq 1 + \sqrt{\frac{J^i(K)}{c_{K,H}}} \leq 2\sqrt{\frac{J^i(K)}{c_{K,H}}}$. 
    % where the last inequality holds since $J^i(K) \geq 1$ from $W = I$. The same argument applies to $\zeta(K)$. 

    (e) observe that $P_K^i = \sum_{t=0}^{H-1}((\bbA_K^i)^t)^\top\bar{P}_K^i(\bbA_K^i)^t$. 
    Then from (b), (c) and the triangle inequality $\norm{P_K^i} \leq H\norm{(\bbA_K^i)^t}^2\norm{\bar{P}_K^i}\leq\frac{H^2J^i(K)^2}{c_{K,H}^2}$. 
    % \[
    %     \norm{P_K^i} \leq H\norm{(\bbA_K^i)^t}^2\norm{\bar{P}_K^i}\leq\frac{H^2J^i(K)^2}{c_{K,H}}
    % \]
\end{proof}

\subsection{Comparing systems and uniform constants on the sublevel set}
\label{s app: uniform bounds}
Based on the above results, we compare the cost of two systems under a heterogeneity bound. 
\begin{lemma}[Cost Bound with Heterogeneity]
    \label{lem app: cost bound with heterogeneity}
    \addcontentsline{toc}{subsubsection}{Lemma \protect\ref{lem: cost bound with heterogeneity}: Cost bound with heterogeneity}
   Consider $\theta_i,{\theta_j}\in\calS$ with optimal FIR controllers $K_i^\star, K_j^\star$ and suppose $\het\leq\barhet^1$ (\Cref{tab:constants}). 
   % \tas{What is $\theta$? Should it be minimum over $\theta_1,\theta_2$} 
   % holds and let $K_i, K_j$ be corresponding optimal controllers for the lifted system. 
   Then $K_i^\star\in\calK^j, K_j^\star\in\calK^i$ and
    \[
        J^i(K_j^\star) \leq 2J^i(K_i^\star), \quad J^j(K_i^\star) \leq 2J^i(K_i^\star)
    \]
\end{lemma}
\begin{proof}
    Since $1\le\norm{\Psta^{i}}, \norm{\Psta^j}\le\Jbs$ and $\norm{P_{K_i^\star}^i} = \norm{\Psta^i}$ (from \Cref{lem app: dare inequalities}), the definition of $\barhet^1$ implies $\het\le\min\curly{1/54\norm{\Psta}^5, 1/30H\norm{P_{K_i^\star}^i}^2, 1}$. 
    
    \textit{First claim.} Note that, from \Cref{lem app: lifted dare}, $K_i^\star$ and $K^\star_j$ have the structured form because of optimality to $i$ and $j$th lifted system $K^\star_i = \begin{bmatrix}
            \Ksta^i & 0 & \cdots & 0
        \end{bmatrix}$
        and 
        $K^\star_j = \begin{bmatrix}
            \Ksta^j & 0 & \cdots & 0
        \end{bmatrix}$.
    Then from \Cref{lem app: policy evaluation with the optimal gain}, $J^i(K^\star_j)$ equals the static cost of the gain $\Ksta^j$ on $\theta_i$. By the certainty-equivalence bound \cite[Thm 5.4]{simchowitz2020naive} (applicable since $\norm{\theta_i-\theta_j}\le\het\le1/54\norm{\Psta}^5$), $K^j$ stabilizes $\theta_i$ and $\Psta^i(\Ksta^j) \preceq \frac{21}{20}\Psta^i(\Ksta^i)$. Thus $J^i(K^\star_j) \le \frac{21}{20}J^i(K^\star_i) \le 2J^i(K^\star_i)$.

    % Then since $\het\le(1/\norm{P_{K_i}^i}^5) = (1/\norm{\Psta^i}^5)$, it follows from \cite[Thm 5.4]{simchowitz2020naive} that
    % \[
    %     \Psta^i(\Ksta^j) \preceq \frac{21}{20}\Psta^i(\Ksta^i)
    % \]
    % Based on above, we can invoke \Cref{lem: policy evaluation with the optimal gain}
    % \begin{align}
    %     P_{K_j}^i &= \diag\paren{\Psta^i(\Ksta^j), ~ p_1Q, \cdots, ~ p_{H-1}Q} \nonumber \\
    %     &\preceq \diag\paren{\frac{21}{20}\Psta^i(\Ksta^i), ~ p_1Q, \cdots, ~ p_{H-1}Q} 
    %     \preceq \frac{21}{20}P^i_{K_i} \label{eq: P bound with heterogeneity}
    % \end{align}
    % where the last inequality follows from \Cref{lem: lifted dare}. 
    % Thus in the cost form, it follows that
    % \[
    %     J^i(K_j) = \trace\paren{P_{K_j}^i\bbW} \leq \frac{21}{20}\trace\paren{P^i_{K_i}\bbW} \le 2J^i(K_i)
    % \]
    % which concludes the part (a).

    \textit{Second claim.} We first compute the difference of Lyapunov matrices. 
    With $\Delta_A\coloneqq\bbA_{K^\star_i}^j - \bbA_{K^\star_i}^i$ and $X\coloneqq (\bbA_{K^\star_i}^i)^\top P_{K^\star_i}^j\Delta_A + \Delta_A^\top P_{K^\star_i}^j\bbA_{K^\star_i}^i + \Delta_A^\top P_{K^\star_i}^j \Delta_A$, it holds that
    \begin{align}
        P_{K^\star_i}^j - P_{K^\star_i}^i = \sum_{l=0}^\infty \paren{(\bbA_{K^\star_i}^i)^l}^\top X(\bbA_{K^\star_i}^i)^l \preceq H\norm{X}\sum_{l=0}^\infty \paren{(\bbA_{K^\star_i}^i)^l}^\top (\bar{\bbQ} + K^\top RK)(\bbA_{K^\star_i}^i)^l = H\norm{X}P^i_{K^\star_i} \label{eq: P_k difference}
    \end{align}
    where the inequality holds since $H(\bar{\bbQ}+K^\top RK)\succeq I$.
    Since $K^\star_i$ is optimal for $i$th lifted system, 
    \[
        \norm{\Delta_A} \leq \norm{A^j - A^i} + \norm{(B^j - B^i)\begin{bmatrix}
            \Ksta^i & 0 & \cdots & 0
        \end{bmatrix}} \leq (1+\norm{\Ksta^i})\het \le (1+\norm{\Psta^i}^{1/2})\het
    \]
    where the last inequality is from \Cref{lem app: dare inequalities} (b). 
    By \Cref{lem app: dare inequalities} (c), $\norm{\bbA_{K^\star_i}^i}\le(\norm{\Psta^i}+1)^{1/2}\le2\norm{\Psta^i}^{1/2}$. 
    % Furthermore, let $A_K = A+BK$ denote the non-lifted closed-loop matrix. Then it follows from the structure of $K_i$ that
    % \begin{align*}
    %     \norm{\bbA_{K_i}^i}^2 &= \lambda_{max}\paren{(\bbA_{K_i}^i)^\top\bbA_{K_i}^i} = \lambda_{max}\brac{\diag\paren{(A_{\Ksta^i}^i)^\top A_{\Ksta^i}^i + I, I, \cdots, I, 0}} \\
    %     &= \max\curly{\norm{A_{\Ksta^i}^i}^2+1, 1} = \norm{A_{\Ksta^i}^i}^2+1 \le \norm{\Psta^i}+1
    % \end{align*}
    Also from \cite[Thm 5.1]{simchowitz2020naive}, $\norm{P^j_{K^\star_j}} = \norm{\Psta^j}\le1.1\norm{\Psta^i} = 1.1\norm{P^i_{K^\star_i}}$. Thus by applying the first claim, we get $\norm{P_{K^\star_i}^j} \leq \frac{21}{20}\norm{P^j_{K^\star_j}} \leq 1.2\norm{P^i_{K^\star_i}}$. Hence, using $\het\le1$, 
    \[
        \norm{X} \le \norm{P_{K^\star_i}^j}\norm{\Delta_A}\paren{2\norm{\bbA_{K^\star_i}^i} + \norm{\Delta_A}} \le 1.2\norm{P^i_{K^\star_i}}\cdot 2\norm{P^i_{K^\star_i}}^{1/2}\het\cdot6\norm{P^i_{K^\star_i}}^{1/2} \le 15\norm{P^i_{K^\star_i}}^2\het.
    \]
    so $H\norm{X}\le1/2$ since $\het\le(1/30H\norm{P_{K^\star_i}^i}^2)$.
    % Since $\het\le(1/20H\norm{P^i}^2)$, it follows that 
    % \[
    %     P_{K_i}^j - P_{K_i}^i \preceq \frac{1}{2}P^i_{K_i}
    % \]
    Finally, the cost error is computed as $J^j(K^\star_i) - J^i(K^\star_i) = \trace\paren{(P_{K^\star_i}^j - P_{K^\star_i}^i)\bbW} \leq \frac{1}{2}\trace\paren{P^i_{K^\star_i}\bbW} = \frac{1}{2}J^i(K_i)$.
\end{proof}

\begin{corollary}[Uniform constants on the initial sublevel set]
    \label{cor:uniform}
    Let $\tJ = 32\Jbs$, $\cb = c(\tJ)$ and $\nb = \nu(\tJ)$ be the $K$-independent constants defined on \Cref{tab:constants}. Suppose $\het\le\min\curly{\barhet^1, \barhet^2}$ (\Cref{tab:constants}). Then it holds that
    \begin{enumerate}[label=(\alph*)]
        \item $\sa(K^\star)\le2\Jbs$ and $K^\star\in\calK_{\zeta_0}$
        \item (scenario boundedness) for every $K\in\calK_{\zeta_0}$ and every $i\in\brac{M}$, $J^i(K)\le2\sa(K)\le\tJ$
        \item for every $K\in\calK_{\zeta_0}$, ~ $c\ge\cb$ and $\nu\le\nb$. 
    \end{enumerate}
\end{corollary}
\begin{proof}
    (a) Fix $\theta\in\calS$. Since $\het\le\barhet^1$, $K(\theta)\in\calK$ and $\sa(K^\star)\le\sa(K(\theta)) = \frac{1}{M}\sum_iJ^i(K(\theta)\le 2J^\theta(K(\theta)) \le 2\Jbs$. Also since $\sa(K^\star)\le\zeta_0$, $K^\star\in\calK_{\zeta_0}$. 

    (b) Let $K\in\calK_{\zeta_0}$ and pick some sample $\Tilde\theta\in\calS$ such that $J^{\tth}(K)\le\sa(K)$. From the initial policy assumption \ref{asmp-app: simultaneous stabilization} and (a), $J^{\tth}(K)\le\sa(K)\le\zeta_0\le8\sa(K^\star)\le16\Jbs\le\tJ$. 
    For any $j\in\brac{M}$, compare $P_K^j$ and $P_K^{\tth}$ via \eqref{eq: P_k difference} (roles $i\rightarrow\tth$), it follows that
    \[
        P_K^{j} - P_K^{\tilde{\theta}} \preceq H\norm{X'}P_K^j, \quad X' \coloneqq (\bbA_K^{\tilde{\theta}})^\top P_{K}^{\tilde{\theta}}\Delta_A' + \Delta_A'^\top P_{K}^{\tilde{\theta}}\bbA_K^{\tilde{\theta}} + \Delta_A'^\top P_{K}^{\tilde{\theta}}\Delta_A', \quad \Delta_A'\coloneqq \bbA_K^j - \bbA_K^{\tilde{\theta}}
    \]
    By invoking \Cref{lem app: cost bound} with $J^{\tth}(K)\le\tJ$ and $c\ge c(\tJ)=\cb$, $\norm{\bbA_K^{\tth}}\le\sqrt{H\tJ/\cb}$, $\norm{P_K^{\tth}}\le H^2\tJ^2/\cb^2$, and $\norm{\Delta_A'}\le\nu(K)\het\le2\sqrt{\tJ/\cb}\het\le\sqrt{H\tJ/\cb}$ since $\het\le\barhet^2\le1/2$. Thus we get 
    \[
        \norm{X'} \le \norm{P_{K}^{\tth}}\norm{\Delta_A'}\paren{2\norm{\bbA_{K}^{\tth}} + \norm{\Delta_A'}} \le \frac{H^2\tJ^2}{\cb^2}\cdot 2\sqrt{\frac{\tJ}{\cb}}\het\cdot3\sqrt{\frac{H\tJ}{\cb}}\le\frac{6H^{5/2}\tJ^3}{\cb^3}\het
    \]
    so $H\norm{X'}\le1/2$ since $\het\le\barhet^2$. Finally, the cost error is computed as $J^j(K) - J^{\tth}(K) = \trace\paren{(P_{K}^j - P_{K}^{\tth})\bbW} \leq \frac{1}{2}\trace\paren{P^j_{K}\bbW} = \frac{1}{2}J^j(K)$, i.e. $J^j(K)\le2J^{\tth}(K)\le2\sa(K)\le32\Jbs=\tJ$.

    (c) follows from (b) and monotonicity of $c(\cdot), \nu(\cdot)$.

\end{proof}

The perturbation argument below are for $K\in\calK_{\zeta_0}$, and therefore from \Cref{cor:uniform} every constant below defined in \Cref{tab:constants} is written as an explicit polynomial in $(H,\tau_B, \tJ, \cb^{-1})$, hence a function of $(H,\tau_B, \Jbs)$ only, independent of $K$. 

\subsection{Perturbation across sytems}
\label{s app: perturbation}
% In this section, we aim to characterize the perturbation of the gradient which is in turn used to bound each system's gradient. 
% We first establish the perturbation bound on multiple system theoretic quantities. 
\begin{lemma}[Perturbation bound]
    \label{lem app: perturbation bound}
    \addcontentsline{toc}{subsubsection}{Lemma \protect\ref{lem: perturbation bound}: Perturbation bound}
    Let $i,j \in \brac{M}$ and $K\in\calK_{\zeta_0}$. 
    % Suppose $\het \leq \barhet^3\triangleq\frac{c_{K,H}^{3/2}}{H^{3/2}\sup_{\theta\in\calS}J^{\theta}(K)}$. 
    Then with $C_A, C_\Sigma, C_P$ from \Cref{tab:constants}, it holds that
    \begin{enumerate}[label=(\alph*)]
        \item $\norm{\bbA_K^i-\bbA_K^j} \leq \nb\het$
        \item $\norm{(\bbA_{K}^i)^t - (\bbA_{K}^j)^t} \leq C_A\nb\het$ for every $0\le t\le H$
        \item $\norm{{\Sigma}_K^i - {\Sigma}_K^j} \leq C_\Sigma\het$
        \item $\norm{\bar{P}_K^i - \bar{P}_K^j}\leq C_P\het$
        % \item $\norm{E_K^i - E_K^j} \leq \paren{\frac{J^i(K)}{c} + \frac{2\tau_BJ^i(K)^{1/2}J^j(K)}{c^{3/2}} + \frac{8\tau_BH^4J^i(K)^2J^j(K)^2}{c^4}}\het$
    \end{enumerate}
\end{lemma}
\begin{proof}
    Throughout, we use $\norm{(\bbA_K^i)^t}, \norm{(\bbA_K^j)^t}\le\sqrt{H\tJ/\cb}$ for all $t\ge0$ from \Cref{lem app: cost bound} with \Cref{cor:uniform}. 
    
    (a) $\|\bbA_K^i - \bbA_K^j\| \leq \norm{A^i - A^j} + \norm{(B^i - B^j)[K_0 ~ K_1 ~ \cdots ~ K_{H-1}]} \leq (1+\norm{K})\het \le \nb\het$.

    (b) Telescoping and using the cost bound:
    \[
        \norm{(\bbA_{K}^i)^t - (\bbA_{K}^j)^t} \leq \norm{\bbA_K^i - \bbA_K^j}\sum_{r=0}^{t-1}\norm{(\bbA_K^i)^{t -1-r}}\norm{(\bbA_K^j)^{r}} \leq \frac{H^2\tJ}{\cb}\nb\het
    \]
    
    (c) Let $D\coloneqq\Sigma_K^i - \Sigma_K^j$ and $D_H\coloneqq \Sigma_H^i(K) - \Sigma_H^j(K)$. 
    Subtracting the two $H$-step Lyapunov equations (\Cref{lem app: strucure of lifted problem} (c)) yields
    % \tas{Why not just use $\Sigma^i_K=A^i_{K}\Sigma^i_K A^i_K+\mathbb W$? Isn't this much easier to analyze?}\Tesshu{It's to invoke $\Sigma_H^j\succeq c_{K,H}I$}
    \[
        D = \underbrace{\bbA_{K,H}^i{\Sigma}_K^i(\bbA_{K,H}^i)^\top - \bbA_{K,H}^j{\Sigma}_K^j(\bbA_{K,H}^j)^\top}_{(i)} + D_H
    \]
    Let $\Delta_{A,H} \coloneqq \bbA_{K,H}^i - \bbA_{K,H}^j$ and $X \coloneqq \bbA_{K,H}^j{\Sigma}_K^i\Delta_{A,H}^\top + \Delta_{A,H}{\Sigma}_K^i(\bbA_{K,H}^i)^\top$. Then since $(i) = \bbA_{K,H}^jD(\bbA_{K,H}^j)^\top + X$, it follows that
    \begin{align}
        \label{eq app: D}
        D 
        = \bbA_{K,H}^jD(\bbA_{K,H}^j)^\top + X + D_H
        = \sum_{t=0}^\infty(\bbA_{K,H}^j)^t(X+D_H)((\bbA_{K,H}^j)^t)^\top \preceq \frac
        {\norm{X} + \norm{D_H}}{\cb}\Sigma_K^j
        % \sum_{t=0}^\infty(\bbA_{K,H}^j)^t(\Sigma_H^j)((\bbA_{K,H}^j)^t)^\top 
        % =  \frac{\norm{X''}\Sigma_K^j}{\cb}
    \end{align}
    where we used $\Sigma_H^j(K) \succeq \cb I$.
    % \begin{align*}
    %     (i) \preceq \frac
    %     {\norm{X}}{c_{K,H}}\sum_{t=0}^\infty(\bbA_{K,H}^j)^t(\Sigma_H^j)((\bbA_{K,H}^j)^t)^\top =  \frac{\norm{X}\Sigma_K^j}{c_{K,H}}
    % \end{align*}
    From \Cref{lem app: cost bound}, $\norm{X}\le2C_A\nb\het\cdot H\tJ\cdot\sqrt{H\tJ/\cb}$.     
    % the norm bound on $X$ is given by
    % \[
    %     \norm{X} \leq 2\paren{\frac{H^{5/2}J^i(K)^{1/2}J^j(K)}{c_{K,H}^{3/2}}\nu(K)\het + \frac{H^4J^i(K)J^j(K)}{c_{K,H}^2}\zeta(K)\het^2}\norm{\Sigma_K^i}
    % \]
    For $D_H$, 
    \begin{align*}
        D_H 
        &= \sum_{t=0}^{H-1}\brac{(\bbA_K^i)^t\bbW\paren{(\bbA_K^i)^t - (\bbA_K^j)^t}^\top + \paren{(\bbA_K^i)^t - (\bbA_K^j)^t}\bbW(\bbA_K^j)^t} 
    \end{align*}
    Thus
    \begin{align}
        \norm{D_H} &\leq H\norm{(\bbA_K^i)^t - (\bbA_K^j)^t}\norm{\bbW}\paren{\norm{(\bbA_K^i)^t} + \norm{(\bbA_K^j)^t}} 
        \le H\cdot C_A\nb\het\cdot2\sqrt{H\tJ/\cb}
        \label{eq app: D_H}
        % \le 4H^{7/2}\tJ^2\cb^{-2}\het
        % \\
        % &\leq \frac{H^3}{c_{K,H}}\sqrt{J^i(K)J^j(K)}\nu(K)\het\cdot\norm{\bbW}\paren{\norm{(\bbA_K^i)^t} + \norm{(\bbA_K^j)^t}} && \because (b) \\
        % &\leq \frac{2H^{7/2}}{c_{K,H}^{3/2}}\sqrt{J^i(K)J^j(K)}\paren{\sqrt{J^i(K)} + \sqrt{J^j(K)}}\nu(K)\het && \because \text{\Cref{lem: cost bound}}
    \end{align}
    Substituting this into \eqref{eq app: D} gives $C_\Sigma$.

    % \begin{align*}
    %     (i) = \bbA_{K,H}^j({\Sigma}_K^i - {\Sigma}_K^j)(\bbA_{K,H}^j)^\top + X = \sum_{t=0}^\infty(\bbA_{K,H}^j)^tX((\bbA_{K,H}^j)^t)^\top
    % \end{align*}
    % where $X = \bbA_{K,H}^j{\Sigma}_K^i(\bbA_{K,H}^i - \bbA_{K,H}^j)^\top + (\bbA_{K,H}^i - \bbA_{K,H}^j){\Sigma}_K^i(\bbA_{K,H}^j)^\top + (\bbA_{K,H}^i - \bbA_{K,H}^j){\Sigma}_K^i(\bbA_{K,H}^i - \bbA_{K,H}^j)^\top$. Since $\Sigma_H^j \succeq cI$,
    % \begin{align*}
    %     {\Sigma}_K^i - {\Sigma}_K^j \preceq \frac
    %     {\norm{X}}{c}\sum_{t=0}^\infty(\bbA_{K,H}^j)^t(\Sigma_H^j)((\bbA_{K,H}^j)^t)^\top =  \frac{\norm{X}\Sigma_K^j}{c}
    % \end{align*}
    % % \tas{this is not $\bar{P}$ since we have $A$ times $A^\top$ not the other way around. Should it involve $\Sigma$ instead. Also I don't get why use $\bar{\Sigma}$ since we have $\Sigma=\bar{\Sigma}$}
    % From \Cref{lem: cost bound}, the norm bound on $X$ is given by
    % \[
    %     \norm{X} \leq 2\paren{\frac{H^{5/2}J^i(K)^{1/2}J^j(K)}{c^{3/2}}\nu(K)\het + \frac{H^4J^i(K)J^j(K)}{c^2}\zeta(K)\het^2}\norm{\Sigma_K^i}
    % \]
    % $\het$ condition yields the result in (c). 

    (d) Identically, let $X' \coloneqq \bbA_{K,H}^{j\top}\bar{P}_K^i\Delta_{A,H} + \Delta_{A,H}^\top\bar{P}_K^i\bbA_{K,H}^i$. Then $\bar{P}_K^i - \bar{P}_K^j = \sum_{l=0}^\infty((\bbA_{K,H}^j)^l)^\top X'(\bbA_{K,H}^j)^l$. 
    % where $X' = (\bbA_{K,H}^j)^\top\bar{P}_K^i(\bbA_{K,H}^i - \bbA_{K,H}^j) + (\bbA_{K,H}^i - \bbA_{K,H}^j)^\top\bar{P}_K^i(\bbA_{K,H}^j) + (\bbA_{K,H}^i - \bbA_{K,H}^j)^\top\bar{P}_K^i(\bbA_{K,H}^i - \bbA_{K,H}^j)$. 
    Since $H(\bar{Q} + K^\top RK) \succeq I$, 
    \begin{align}
        \label{eq app: P perturbation}
        \bar{P}_K^i - \bar{P}_K^j \preceq H\norm{X'}\sum_{l=0}^\infty((\bbA_{K,H}^j)^l)^\top (\bar{\bbQ} + K^\top RK)(\bbA_{K,H}^j)^l = H\norm{X'}\bar{P}_K^j
    \end{align}
    By \Cref{lem app: cost bound}, $\norm{X'}\le2C_A\nb\het\cdot\sqrt{H\tJ/\bc}\cdot\tJ/\bc$. Substituting this into \eqref{eq app: P perturbation} with \Cref{lem app: cost bound} (b) yields $C_P$. 
    % \[
    %     \norm{X'} \leq 2\paren{\frac{H^{5/2}J^i(K)^{1/2}J^j(K)}{c_{K,H}^{3/2}}\nu(K)\het + \frac{H^4J^i(K)J^j(K)}{c_{K,H}^2}\zeta(K)\het^2}\norm{\bar{P}_K^i}
    % \]
    % which yields the result in (d).
\end{proof}

% Now we reparameterize gradient $\nabla_K\sa(K)$ in \eqref{eq: gradient} with $\bar{P}_K^i$ and instantiate both the cost and perturbation bounds from previous lemmas to obtain gradient perturbation. 
\begin{lemma}[$E_K$ Perturbation]
    \label{lem app: E_K perturbation}
    \addcontentsline{toc}{subsubsection}{Lemma \protect\ref{lem: E_K perturbation}: $E_K$ perturbation}
    % Suppose $\het \leq \barhet^3$ which is defined in \Cref{lem: perturbation bound}. 
    Let $K\in\calK_{\zeta_0}$ and $i,j\in\brac{M}$. 
    Then $E_K^i$ admits the representation
    \[
        E_K^i = RK + \sum_{t=0}^{H-1}\bbB^{i\top}(\bbA_K^i)^{t\top}\bar{P}_K^i(\bbA_K^i)^{t+1}.
        % \tas{\text{if this is equal to $E^i_K$ let's not introduce new notation}}\Tesshu{first explain E_K=\bar{E}_K}
    \]
    and $\norm{E_K^i - E_K^j} \leq C_E\het$
    with $C_E$ from \Cref{tab:constants}. 
\end{lemma}
\begin{proof}
    From \Cref{lem app: strucure of lifted problem} (c), $P_K^i = \sum_{t=0}^{H-1}(\bbA_K^i)^{t\top}\bar{P}_K^i(\bbA_K^i)^t$.
    Applying this to $E_K^i$ yields
    \[
        E_K^i = RK + \bbB^{i\top} P_K^i \bbA_K^i = RK + \sum_{t=0}^{H-1}\bbB^{i\top}(\bbA_K^i)^{t\top}\bar{P}_K^i(\bbA_K^i)^{t+1}
    \]
    Let $\Tilde{P}_t^i \triangleq (\bbA_K^i)^{t\top}\bar{P}_K^i(\bbA_K^i)^{t+1}$, so that $\|\tilde{P}_t^i\|\le\|(\bbA_K^i)^t\|\|(\bbA_K^i)^{t+1}\|\|\bar{P}_K^i\| \leq H\tJ^2/\cb^2$. By \Cref{lem app: perturbation bound} (b) (d), 
    \[
        \norm{\Tilde{P}_t^i - \Tilde{P}_t^j} \le C_A\nb\het\cdot(\|\bar{P}_K^i(\bbA_K^i)^{t+1}\|+\|(\bbA_K^j)^t\bar{P}_K^j\|) + C_P\het\cdot\|(\bbA_K^i)^{t+1}\|\|(\bbA_K^j)^t\| \le 12H^{9/2}\tJ^5\cb^{-5}\het
    \]
    Then it follows that $\|\bbB^{i\top}\tilde{P}_t^i - \bbB^{j\top}\tilde{P}_t^j\| \le \het\|\tilde{P}_t^i\| + \|\bbB^{j\top}\|\|\tilde{P}_t^i - \tilde{P}_t^j\| \le 13\tau_BH^{9/2}\tJ^5\cb^{-5}\het$. 
    Since the $RK$ terms cancel, by summing over $0\le t<H$ gives $\|E_K^i - E_K^j\| \le H\|\bbB^{i\top}\tilde{P}_t^i - \bbB^{j\top}\tilde{P}_t^j\| \le C_E\het$.     
\end{proof}

\begin{lemma}[Gradient Perturbation]
    \label{lem app: gradient perturbation}
    \addcontentsline{toc}{subsubsection}{Lemma \protect\ref{lem: gradient perturbation}: Gradient Perturbation}
    Let $K\in\calK_{\zeta_0}$. For every $i\in\brac{M}$, 
    \begin{align*}
        \sum_{j=1,j\neq i}^M\norm{\nabla J^i(K) - \nabla J^j(K)}_F \leq M\paren{C_{G,1} + C_{G,2}\norm{\nabla J^i(K)}_F}\het
    \end{align*}
    with $C_{G,1}=2H\tJ C_E$ and $C_{G,2} = 2C_\Sigma/\cb$ as in \Cref{tab:constants}. 
    % where
    % \begin{align*}
    %     M_1^i &= \frac{96\tau_BM\bfs^4H^{13/2}\sa(K)^5}{c_{K,H}^5}J^i(K),\qquad M_2 = \frac{12M\bfs H^{7/2}\sa(K)^2}{c_{K,H}^3}J^i(K)
    % \end{align*}
\end{lemma}
\begin{proof}
    $\frac{1}{2}(\nabla J^i - \nabla J^j) = E_K^i(\Sigma_K^i - \Sigma_K^j) + (E_K^i-E_K^j)\Sigma_K^j$ and $E_K^i\Sigma_K^i = \frac{1}{2}\nabla J^i$. Hence
    \begin{align*}
        \norm{\nabla J^i - \nabla J^j}_F \le \norm{\nabla  J^i}_F\norm{(\Sigma_K^i)^{-1}(\Sigma_K^i - \Sigma_K^j)} + 2\norm{E_K^i - E_K^j}\norm{\Sigma_K^j}_F \le \paren{\frac{C_\Sigma}{\cb}\norm{\nabla J^i}_F + 2H\tJ\cdot C_E}\het
    \end{align*}
    using $\lmin(\Sigma_K^i)\ge\cb$, \Cref{lem app: cost bound} (a), \Cref{lem app: perturbation bound} (c) and \Cref{lem app: E_K perturbation}. Summing over $j\neq i$ concludes the proof. 
    % To bound $(i)$, observe that $\Sigma_K^i - \Sigma_K^j$ can be bounded as follows from \Cref{lem: perturbation bound}
    % \[
    %     \Sigma_K^i - \Sigma_K^j \preceq \alpha(K)\Sigma_K^i + (\Sigma_H^i(K) - \Sigma_H^j(K)) \preceq \brac{\alpha(K) + \beta(K)}\Sigma_K^i
    % \]
    % where $\alpha(K) = \paren{\frac{8H^{7/2}}{c_{K,H}^3}J^i(K)J^j(K)^2\het}$ and $\beta(K) = \norm{\Sigma_H^i(K) - \Sigma_H^j(K)}/c_{K,H}$. 
    % Thus we get
    % \[
    %     (i) = \sqrt{\norm{(\Sigma_K^i)^{-1}(\Sigma_K^i - \Sigma_K^j)^2(\Sigma_K^i)^{-1}}} = \alpha(K) + \beta(K)
    % \]
    % With \Cref{lem: cost bound} and \Cref{lem: E_K perturbation}, it follows that
    % \begin{align*}
    %     \eqref{eq: gradient diff} \leq \frac{12H^{7/2}}{c^3}J^i(K)J^j(K)^2\norm{\nabla_K J^i(K)}_F\het + \frac{96\tau_BH^{13/2}}{c_{K,H}^5}J^i(K)^{5/2}J^j(K)^{7/2}\het
    % \end{align*}    
    % \begin{align*}
    %      \norm{\nabla J^i - \nabla J^j} \leq& HJ^j\paren{\frac{J^i(K)}{c} + \frac{2\tau_BJ^i(K)^{1/2}J^j(K)}{c^{3/2}} + \frac{8\tau_BH^4J^i(K)^2J^j(K)^2}{c^4}}\het \\
    %      &+ \norm{\nabla J^i}\frac{8H^{7/2}}{c^3}J^i(K)J^j(K)^2\het
    % \end{align*}
    % Applying $J^i(K)\leq\bfs\sa(K)$ to $i$th and $j$ cost and summing up from $j = 1$ to $M$ concludes the proof. 
\end{proof}

\begin{lemma}[Gradient bound]
    \label{lem app: gradient bound}
    \addcontentsline{toc}{subsubsection}{Lemma \protect\ref{lem: gradient bound}: Gradient bound}
    Suppose $\het\leq\barhet^3\coloneqq (2C_{G,2})^{-1}$ (\Cref{tab:constants}).  
    % \tas{I am quite confused that the definition of $\grad$ depends on $K$. The results seems a bit cyclical since both sides depend on $K$}
    % /$\grad \triangleq \norm{\nabla_K \sa(K)}_F$. 
    Then for every $K\in\calK_{\zeta_0}$ and $i\in\brac{M}$, 
    \begin{align}
        \label{eq: grad-i-ub}
        \norm{\nabla J^i(K)}_F \leq \Lambda(K) \coloneqq 2\grad(K) + \Lambda_1\het
    \end{align}
    with $\Lambda_1 = 2C_{G,1}$ (\Cref{tab:constants}). 
    Furthermore, it holds that
    \begin{align}
        \frac{1}{M}\sum_{i=1}^M\norm{\nabla J^i(K)}_F \leq \Lambda(K). \label{eq: Lambda}
    \end{align}
\end{lemma}
\begin{proof}
$\nabla J^i = \nabla \sa + \frac{1}{M}\sum_{j\neq i}(\nabla J^i - \nabla J^j)$. Hence from \Cref{lem app: gradient perturbation}, $\|\nabla J^i\|_F \le \grad(K) + (C_{G,1} + C_{G,2}\|\nabla J^i\|)\het$. By the choice of $\barhet^3$, $C_{G,2}\het\le 1/2$. Thus by rearranging, we get \eqref{eq: grad-i-ub}. 
% \begin{align*}
%     \norm{\nabla_K J^i(K)}_F &= \norm{\nabla_K J^i(K) + \frac{1}{M}\sum_{j=1, j\neq i}^M\paren{\nabla_K J^j(K) - \nabla_K J^j(K)}}_F \\
%     &= \norm{\frac{1}{M}\sum_{i=1}^M\nabla_KJ^i(K) + \frac{1}{M}\sum_{j=1, j\neq i}^M\paren{\nabla_K J^i(K) - \nabla_K J^j(K)}}_F \\
%     &\leq \norm{\nabla_K\sa(K) + \frac{1}{M}\sum_{j=1, j\neq i}^M\paren{\nabla_K J^i(K) - \nabla_K J^j(K)}}_F \\
%     &\leq \frac{1}{M}\paren{M \grad + \paren{M^i_1 + M_2\norm{\nabla_KJ^i(K)}_F}\het}.
% \end{align*}
% Note that the second term is characterized by a function of the heterogeneity with the coefficients defined in \Cref{lem: gradient perturbation}.
% \tas{in many parts I see both $\nabla$ and $\nabla_K$. Polish notation}
% By grouping the term, we get
% \begin{align*}
%     \norm{\nabla J^i(K)}_F \leq \frac{M \grad + M^i_1\het}{M - M_2\het} \leq \frac{2}{M}\paren{M \grad + M^i_1\het}.
% \end{align*}
% which matches \eqref{eq: grad-i-ub}.
% The second inequality in the above bound follows from the fact that for $\het \leq \frac{c_{K,H}^3}{24H^{7/2}\bfs^2\sa(K)^3}$, 
% \begin{align*}
%     &M_2\het < \frac{M}{2}.
% \end{align*}
Taking the average of \eqref{eq: grad-i-ub} from $i = 1$ to $M$ yields \eqref{eq: Lambda}.
\end{proof}

\subsection{Perturbation across gains}
\label{s app: perturbation across gains}
\begin{lemma}[Gain perturbation]
    \label{lem app: gain perturbation}
    Let $K, K'\in\calK_{\zeta_0}$, $\Delta\coloneqq K'-K$, and $i\in\brac{M}$. Then with $C_A, C_\gamma$ from \Cref{tab:constants}, it holds that
    \begin{enumerate}[label=(\alph*)]
        \item $\norm{(\bbA_{K'}^i)^t - (\bbA_K^i)^t}\le\tau_BC_A\norm{\Delta}_F$ for $0\le t\le H$
        \item $\norm{(\Sigma_K^i)^{-1}(\Sigma_{K'}^i - \Sigma_K^i)}\le C_\gamma\norm{\Delta}_F$
        \item $\norm{\Sigma_{K'}^i - \Sigma_K^i}_F\le H\tJ C_\gamma\norm{\Delta}_F$.
    \end{enumerate}
\end{lemma}
\begin{proof}
    (a) $\bbA_{K'}^i - \bbA_K^i = \bbB^i\Delta$, thus $\|\bbA_{K'}^i - \bbA_K^i\|\le\tau_B\|\Delta\|_F$. Telescoping as in \Cref{lem app: perturbation bound} (b) gives the bound.

    (b) Let $D\coloneqq\Sigma_{K'} - \Sigma_K$ and $\Delta_{A,H} \coloneqq \bbA_{K',H}^i - \bbA_{K,H}^i$. Unrolling along $\bbA_{K'}^i$ as in \Cref{lem app: perturbation bound} (c), \eqref{eq app: D} and using $\|\Sigma_{K'}^i\|\le H\tJ$, 
    \[
        \norm{D}\le \paren{\norm{X} + \norm{\Sigma_H^i(K') - \Sigma_H^i(K)}}\frac{H\tJ}{\cb}, \quad X\coloneqq \bbA_{K,H}^i\Sigma_K^i\Delta_{A,H}^\top + \Delta_{A,H}^\top\Sigma_K^i\bbA_{K',H}^i.
    \]
    From (a) and \Cref{lem app: cost bound}, $\|X\|\le 2\tau_BC_A\|\Delta\|_F\cdot H\tJ\cdot \sqrt{H\tJ/\cb}$. Also, from (a) and the analogy to \eqref{eq app: D_H}, $\|\Sigma_H^i(K') - \Sigma_H^i(K)\| \le H\cdot \tau_BC_A\|\Delta\|_F\cdot 2\sqrt{H\tJ/\cb}$. Finally, $\|(\Sigma_K^i)^{-1} D\|\le\|D\|/\cb$ gives $C_\gamma$ in \Cref{tab:constants}. 

    (c) $\|D\|_F = \|\Sigma_K^i[(\Sigma_K^i)^{-1} D]\|_F \le \|\Sigma_K^i\|_F\|(\Sigma_K^i)^{-1} D\|$ by \Cref{lem app: cost bound} gives the bound. 
\end{proof}

\section{Gradient domination and global convergence}
\label{s app: gradient domination and global convergence}
Throughout this section, Assumptions \ref{asmp-app: simultaneous stabilization}-\ref{asmp-app: heterogeneity assumption} hold. In particular, $\het\le\min\curly{\barhet^1, \cdots, \barhet^4}$ (\Cref{tab:constants}), so \Cref{cor:uniform} and \Cref{lem app: gradient bound} are valid for every $K\in\calK_{\zeta_0}$. Also, recall $K^\star\in\calK_{\zeta_0}$ and $\nabla\sa(K^\star) = 0$ (\Cref{lem app: coercivity of sa FIR-LQR}).

\begin{lemma}[Approximate gradient domination, adapted from {\cite[{Lemma~III.3}]{fujinami2025policygradientlqrdomain}}]
    \label{lem app: approximate gradient domination}
    For every $K\in\calK_{\zeta_0}$, 
    \[
        \sa(K) - \sa(K^\star) \le \frac{1}{2\cb}\norm{\nabla\sa(K)}_F^2 + \frac{2}{\cb}\norm{\sfR(K,K^\star)}_F^2, \quad \sfR(K,K^\star)\coloneqq \frac{1}{M}\sum_{i=1}^ME_K^i(\Sigma_{K^\star}^i - \Sigma_K^i)
    \]
\end{lemma}
\begin{proof}
    Following the proof of \cite[Lemma~III.3]{fujinami2025policygradientlqrdomain} with $\lmin(\Sigma_{K^\star}^i)\ge\cb$ (\Cref{cor:uniform}), 
    \[
        \sa(K) - \sa(K^\star) \le \frac{1}{\cb}\norm{T}_F^2, \quad T\coloneqq \frac{1}{M}\sum_{i=1}^ME_K^i\Sigma_{K^\star}^i. 
    \]
    Since $T = \frac{1}{M}\sum_{i=1}^ME_K^i\Sigma_{K}^i + \sfR(K,K^\star) = \frac{1}{2}\nabla\sa(K) + \sfR(K,K^\star)$, $\|T\|_F^2\le\frac{1}{2}\|\nabla\sa(K)\|_F^2 + 2\|\sfR\|_F^2$. 
\end{proof}

\begin{lemma}[Quadratic growth of $\sa$]
    \label{lem app: quadratic growth of sa}
    Suppose $\het\le\barhet^4$ (\Cref{tab:constants}). Then for every $K\in\calK_{\zeta_0}$, 
    \begin{align}
        \sa(K) - \sa(K^\star) \ge \frac{\cb}{2}\norm{K-K^\star}_F^2.
        \label{eq app: quadratic growth of sa}
    \end{align}
\end{lemma}
\begin{proof}
    Let $\Delta\coloneqq K - K^\star$. \eqref{eq app: cost difference} with base point $K^\star$ averaged over $i\in\brac{M}$ gives
    \[
        \sa(K) - \sa(K^\star) = S + \frac{1}{M}\sum_{i=1}^M\trace\paren{\Sigma_K^i\Delta^\top(R + \bbB^{i^\top}P_{K^\star}^i\bbB^i)\Delta} \ge S + \cb\norm{\Delta}_F^2, \quad S \coloneqq \frac{2}{M}\sum_{i}\trace\paren{\Sigma_K^i\Delta^\top E_{K^\star}^i}
    \]
    where the quadratic term is bound using $R = I$ and $\lmin(\Sigma_K^i)\ge\cb$ from \Cref{cor:uniform}. Now note that $\nabla\sa(K^\star) = \frac{2}{M}\sum_{i=1}^ME_{K^\star}^i\Sigma_{K^\star}^i = 0$. Then using Cauchy-Schwarz inequality, 
    \[
        S = \frac{2}{M}\sum_{i}\trace\paren{(\Sigma_K^i - \Sigma_{K^\star}^i)\Delta^\top E_{K^\star}^i} \ge -\max_i\norm{E_{K^\star}^i}\cdot\frac{2}{M}\sum_{i}\norm{\Sigma_K^i - \Sigma_{K^\star}^i}_F\norm{\Delta}_F.
    \]
    By \Cref{lem app: gradient bound}, $\grad(K^\star)=0$, and $E_{K^\star}^i = \frac{1}{2}\nabla J^i(K^\star)(\Sigma_{K^\star}^i)^{-1}$; $\|\nabla J^i(K^\star)\|_F\le \Lambda_1\het$ and $\|E_{K^\star}^i\|\le\Lambda_1\het/2\cb$. Also, by \Cref{lem app: gain perturbation} (c), $\|\Sigma_K^i-\Sigma_{K^\star}^i\|_F\le H\tJ C_\gamma\|\Delta\|_F$. Hence $S\ge-(\Lambda_1H\tJ C_\gamma /\cb)\het\|\Delta\|_F^2 \ge -\cb/2\|\Delta\|_F^2$ by the choice of $\barhet^4$ (\Cref{tab:constants}), and \eqref{eq app: quadratic growth of sa} follows. 
\end{proof}

\begin{lemma}[Gradient domination of $\sa$]
    \label{lem app: gradient domination of sa}
    Define
    \begin{align}
        &\bargrad \coloneqq \frac{\cb}{8C_\gamma}, \quad \fgrad \coloneqq \frac{1}{\bargrad} \label{eq app: bargrad} \\
        &\barhet \coloneqq \min\curly{\barhet^1, \barhet^2, \barhet^3, \barhet^4, \frac{\cb}{4C_\gamma\Lambda_1}}, \quad \fhet \coloneq \frac{1}{\barhet} \label{eq app: barhet}
    \end{align}
    with all constants from \Cref{tab:constants}. Then $\fgrad, \fhet$ are explicit, finite, positive functions of $(H,\tau_B,\Jbs)$, and under Assumptions \ref{asmp-app: simultaneous stabilization}-\ref{asmp-app: heterogeneity assumption}, every $K\in\calK_{\zeta_0}$ with $\grad(K)\le\bargrad$ satisfies
    \begin{align}
        \sa(K) - \sa(K^\star) \le \frac{1}{\cb}\norm{\nabla\sa(K)}_F^2. \label{eq app: gradient domination of sa}
    \end{align}
\end{lemma}
\begin{proof}
    Let $\Delta\coloneqq K-K^\star$ and $\delta \coloneqq \sa(K) - \sa(K^\star)$. From $\nabla J^i(K) = 2E_K^i\Sigma_K^i$, \Cref{lem app: approximate gradient domination}, \Cref{lem app: gain perturbation} (b) and \Cref{lem app: gradient bound}, 
    \[
        \norm{\sfR(K,K^\star)}_F = \norm{\frac{1}{M}\sum_i\frac{1}{2}\nabla J^i(K)(\Sigma_K^i)^{-1}(\Sigma_{K^\star}^i - \Sigma_K^i)_F} \le \frac{1}{2M}\sum_i\norm{\nabla J^i(K)}_F\cdot C_\gamma\norm{\Delta}_F \le \frac{C_\gamma\Lambda(K)}{2}\norm{\Delta}_F
    \]
    Quadratic growth from \Cref{lem app: quadratic growth of sa} gives
    \[
        \norm{\sfR(K,K^\star)}_F^2\le\frac{C_\gamma^2\Lambda(K)^2}{4}\cdot\frac{2\delta}{\cb} = \frac{C_\gamma^2\Lambda(K)^2}{2\cb}\delta
    \]
    By \eqref{eq app: bargrad}-\eqref{eq app: barhet} and \Cref{lem app: gradient bound}, $\Lambda(K) = 2\grad(K) + \Lambda_1\het \le \frac{\cb}{4C_\gamma} + \frac{\cb}{4C_\gamma} = \frac{\cb}{2C_\gamma}$, so $\|\sfR\|_F^2\le\frac{\cb}{8}\delta$. 
    Substituting into approximate gradient domination (\Cref{lem app: approximate gradient domination}), $\delta\le\frac{1}{2\cb}\|\nabla\sa(K)\|_F^2 + \frac{2}{\cb}\cdot\frac{\cb\delta}{8}$, which rearranges to $\delta\le\frac{2}{3\cb}\|\nabla\sa(K)\|_F^2 \le \frac{1}{\cb}\|\nabla\sa(K)\|_F^2$. 
\end{proof}

Finally, we show that policy gradient with a suitable stepsize converges globally to a minimizer of $\sa(K)$. 
\begin{theorem}[Global convergence of sample average FIR-LQR]
    \label{thm app: global convergence}
    Suppose Assumptions \ref{asmp-app: simultaneous stabilization}-\ref{asmp-app: heterogeneity assumption} hold.  
    Let $L$ be the $L$-smoothness constant
    % \tas{what is the smoothness constant? Lipschitz constant? L-smooth constant? can it be found somewhere? you could also say there exists a step-size such that... } 
    of $\sa(K)$ over $K_{\zeta_0}$, and $\grad(K) \le \bargrad = 1/\fgrad = \cb/8C_\gamma$ as in \eqref{eq app: bargrad}.
    Then policy gradient \eqref{eq app: pg} over the class of FIR policies of order $H$  with stepsize $\alpha = 1/L$ achieves $\sa(K^{(N)})
    - \sa(K^\star)\leq\epsilon$ after 
    % \tas{it seems somehow that $\varepsilon_{grad}$ affects this result, I cannot tell where it should enter, shouldn't you further restrict $K^0$ to be in $S$?}
    \[
        N\geq 2L\zeta_0\max\curly{\frac{1}{\cb\epsilon}, \fgrad^2}
    \]
    iterations.
\end{theorem}
\begin{proof}
    From \Cref{lem app: convergence to a fixed point}, every iterate satisfies $K^{(n)}\in\calK_{\zeta_0}$, and monotonicity $\sa(K^{(n+1)})\le\sa(K^{(n)})$. With $\alpha = 1/L$, $C = \alpha - L\alpha^2/2 = 1/2L$, and thus \eqref{eq app: grad convergence rate} reads
    \begin{align}
        \min_{0\le n\le N}\norm{\nabla\sa(K^{(n)})}_F^2 \le \frac{2L\zeta_0}{N+1}. \label{eq app: gradient bound}
    \end{align}
    Let $l^\star$ attain the minimum of \eqref{eq app: gradient bound}. We first show that $\sa(K^{(l)})$ is $\epsilon$-optimal to the optimal solution. Then we claim that the same result holds for the $N$-th SA objective from monotonicity. 

    When $N\ge2L\zeta_0\fgrad^2$, it holds from \eqref{eq app: gradient bound} that
    \[
        \norm{\nabla\sa(K^{(l^\star)})}_F \le \frac{1}{\fgrad} = \bargrad.
    \]
    As Assumptions \ref{asmp-app: simultaneous stabilization}, \ref{asmp-app: heterogeneity assumption} hold and $K^{l^\star}\in\calK_{\zeta_0}$, \Cref{lem app: gradient domination of sa} applies at $K^{(l^\star)}$ and gives
    \[
        \sa(K^{l^\star}) - \sa(K^\star) \le \frac{1}{\cb}\norm{\sa(K^{(l^\star)})}_F^2 \le \frac{2L\zeta_0}{N+1}\le\epsilon
    \]
    after $N\ge2L\zeta_0/\cb\epsilon$ iterations. 
    Finally from monotonicity, $\sa(K^{(N)}) - \sa(K^\star) \le \sa(K^{l^\star}) - \sa(K^\star) \le \epsilon$. 
\end{proof}

\section{Further experiments in \Cref{s: alternative dynamic}}
\label{s app: further experiments in dynamic form}
Another useful parameterization is the canonical form controller
\begin{equation}
    \label{eq: canonical}
    \begin{split}
        &A_K = \begin{bmatrix} -d_3 & -d_2 & -d_1 & -d_0 \\ 1 & 0 & 0 & 0 \\ 0 & 1 & 0 & 0 \\ 0 & 0 & 1 & 0 \end{bmatrix}, \quad B_K = \begin{bmatrix} 1 \\ 0 \\ 0 \\ 0 \end{bmatrix}, \\
        &C_K = \begin{bmatrix} n_3 & n_2 & n_1 & n_0 \end{bmatrix}, \quad D_K = 0
    \end{split}
\end{equation}

As shown in \Cref{fig: canonical}, the canonical form is also effective for identifying a simultaneously stabilizing controller. 
Since it has more parameters than the FIR class, it only uses $H=2$ history while optimization remains stable. This further supports the tradeoff that limiting the degree of freedom of the controller class can simplify policy optimization while still achieving strong performance. However, compared with FIR controller in \Cref{fig: fir}, the canonical class attains a suboptimal solution. One potential explanation is that, by limiting $D_K = 0$,  the canonical form may fail to represent the best controller.

\begin{figure}[tbhp]
    \centering
    \includegraphics[width=0.7\linewidth]{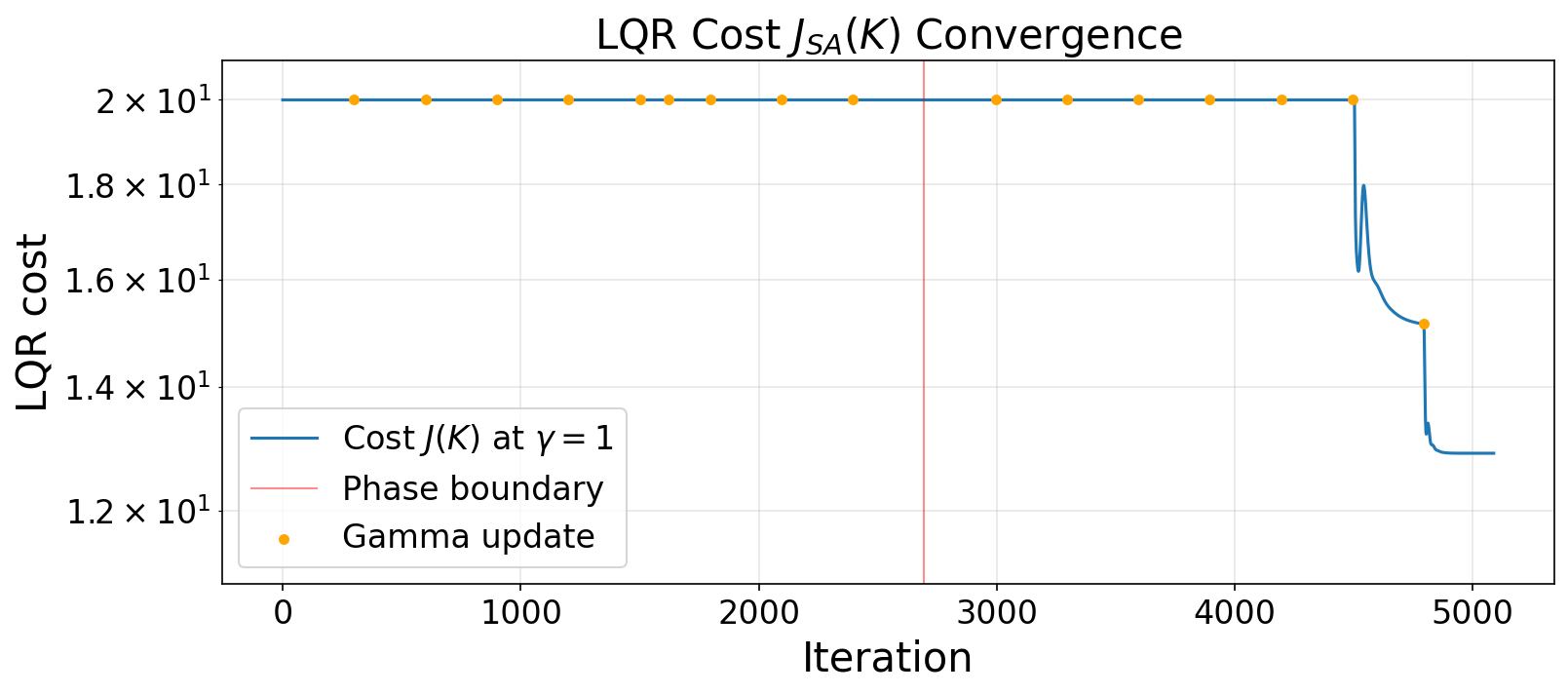}
    \caption{Policy gradient on Example \ref{ex:scalar-lqr-unstable} with a canonical form controller \eqref{eq: canonical}.}
    \label{fig: canonical}
\end{figure}

\end{document}